\documentclass[sigconf]{acmart}

\pdfoutput=1 

\let\ACMmaketitle=\maketitle
\renewcommand{\maketitle}{\begingroup\let\footnote=\thanks 
	\ACMmaketitle\endgroup}

\renewcommand\footnotetextcopyrightpermission[1]{}

\newif\ifcomments
\commentstrue

\AtBeginDocument{%
  \providecommand\BibTeX{{%
    \normalfont B\kern-0.5em{\scshape i\kern-0.25em b}\kern-0.8em\TeX}}}

\newif\iftr
\newif\ifconf
\conftrue
\trfalse

\usepackage{soul}
\usepackage{color, colortbl}
\usepackage{moresize}

\usepackage{rotating}
\usepackage{graphicx}
\usepackage{multirow}
\usepackage{amsfonts}
\usepackage{amsmath}

\usepackage{mathtools}
\usepackage{wrapfig}
\usepackage{listings}
\usepackage{graphicx}
\usepackage{caption}
\usepackage{tabularx}
\usepackage{fontawesome}
\usepackage{pifont}
\usepackage{enumitem}
\usepackage[font={small}]{caption}
\usepackage{algpseudocode}
\usepackage[font={bf,sf,scriptsize}]{subfig} 
\usepackage{amsthm}
\usepackage{placeins}
\usepackage{bm}
\usepackage[framemethod=tikz]{mdframed}
\usepackage{algorithm}
\usepackage{algpseudocode}
\usepackage{algorithmicx}
\usepackage{verbatim}
\usepackage{datetime}
\usepackage{xr}
\usepackage{float}
\usepackage{thmtools} 
\usepackage{thm-restate}
\usepackage[font={bf,sf,footnotesize}]{caption}
\usepackage{cleveref}
\usepackage{xspace}
\usepackage[utf8]{inputenc}
\usepackage{booktabs} 
\usepackage{makecell}
\usepackage{pifont}
\usepackage{etex}
\usepackage{tikz}
\usetikzlibrary{positioning,fit,calc,shapes.geometric}
\usepackage{soul} 
\usepackage{hhline}
\usepackage{appendix}
\usepackage{dblfloatfix}

\crefname{section}{§}{§§}
\Crefname{section}{§}{§§}

\newfloat{lstfloat}{htbp}{lop}
\floatname{lstfloat}{Listing}

\newtheorem*{corollary*}{Corollary}

\newtheorem{lma}{Lemma}
\newtheorem{defn}{Definition}
\newtheorem{thm}{Theorem}

\newtheorem{crl}{Corollary}
\newtheorem{eg}{Example}

\newtheorem{observation}{Observation}

\newcommand{\xparting}{\mbox{$X$-Partitioning}\xspace}
\newcommand{\xpart}{\mbox{$X$-partition}\xspace}

\definecolor{darkgreen}{rgb}{0.0, 0.5, 0.13}

\DeclareMathOperator*{\argmin}{arg\,min}

\newcommand\macsection[1]{\noindent \textbf{#1. }}

\usepackage{xcolor}
\colorlet{hlcolor}{yellow!20}

\acmConference[Technical Report]{Technical Report}{2021}{}
\startPage{1}

\setcopyright{none}

\begin{document}
	\title[Tight I/O Bounds of Statically Analyzable Programs]{Pebbles, Graphs, 
	and a Pinch of Combinatorics: 
	Towards Tight I/O 
	Lower Bounds for Statically Analyzable Programs}
	
	\author{Grzegorz Kwasniewski, Tal Ben-Nun, Lukas 
	Gianinazzi, Alexandru Calotoiu, Timo Schneider, 
	Alexandros Nikolaos Ziogas, Maciej Besta, Torsten Hoefler}
	\affiliation{%
				\country{ETH Zurich, Switzerland}
			}

\renewcommand{\shortauthors}{G. Kwasniewski et al.}

\begin{abstract}
Determining I/O lower bounds is a crucial step in obtaining com-munication-efficient
parallel algorithms, both across the memory hierarchy and between processors.
Current approaches either study specific algorithms individually, disallow programmatic motifs such as recomputation, or produce asymptotic bounds that exclude important constants.
We propose a novel approach for obtaining precise I/O lower bounds on a general class of programs, which we call Simple Overlap Access Programs (SOAP).
SOAP analysis covers a wide variety of algorithms, from ubiquitous computational kernels to full scientific computing applications.
Using the red-blue pebble game and combinatorial methods, we 
are able to bound the I/O of the SOAP-induced Computational 
Directed Acyclic Graph (CDAG), taking into account multiple 
statements, input/output reuse, and optimal tiling.
To deal with programs that are outside of our representation (e.g., non-injective access functions), we describe methods to approximate them with SOAP.
To demonstrate our method, we analyze 38 different applications, including kernels from the Polybench benchmark suite, deep learning operators, and --- for the first time --- applications in unstructured physics simulations, numerical weather prediction stencil compositions, and full deep neural networks.
We derive tight I/O bounds for several linear algebra kernels, such as Cholesky 
decomposition, improving the existing reported bounds by a factor of two. For 
stencil applications, we improve the existing bounds by a factor of up to 14.
We implement our method as an open-source tool, which can derive lower bounds directly from provided C code.

\end{abstract}
	
	\vspace{-0.5em}


\vspace{-0.5em}
	
	\maketitle

		\vspace{-1em}
	\section{Introduction}
	
I/O operations, both across the memory hierarchy and between parallel
processors, dominate time and energy costs in many scientific
applications~\cite{survey, kestor2013quantifying, tate2014programming, padal}.
It is thus of key importance to design algorithms with communication-avoiding or 
\emph{I/O-efficient} schedules~\cite{maciejBC,
edgarTradeoff}.
To inform, and occasionally inspire the development of such algorithms, one must first understand
\emph{the associated lower bounds on the amounts of communicated data}.
Deriving these bounds has always been of theoretical
interest~\cite{general_arrays, redblue}.
It is particularly relevant for dense linear algebra, as many
important problems in scientific computing~\cite{joost, rectangularML} and machine learning~\cite{ddl} rely on linear algebra operations such as matrix
factorization~\cite{meyer2000matrix, krishnamoorthy2013matrix} or tensor
contractions~\cite{solomonik2014massively}. 

Analyzing I/O bounds of linear algebra kernels dates back to the seminal work by
Hong and Kung{~\cite{redblue}}, who derived the first asymptotic bound for
matrix-matrix multiplication (MMM) using the red-blue pebble game abstraction.
This method was subsequently extended and used by other works to derive
asymptotic{~\cite{ElangoSymbolic}} and tight{~\cite{COSMA}} bounds for more
complex programs. 
Despite the expressiveness of pebbling, it is prohibitively hard to
solve for arbitrary programs, as it is PSPACE-complete in the general case~\cite{redbluehardSPAA}.

Since analyzing programs with parametric sizes disallows the  
construction of an explicit Computation Directed Acyclic Graph (CDAG), some 
form of parameterization is often needed~\cite{redbluewhite, olivry2020automated, chineseCNN1}. However, we argue that the widely-used 
approaches based on the Loomis-Whitney or the HBL 
inequalities~\cite{demmel1,demmelHBL,demmelCNN} (a) are often too restrictive, 
requiring the programs to be expressed in the polyhedral model to count the 
points in the projection polytopes; (b) do not capture pebbling motifs such as recomputation~\cite{olivry2020automated}; or (c) are limited to single-statement
programs~\cite{general_arrays, demmelCNN, demmel1, demmel2, 
demmelCNN, demmelHBL}.

	\begin{figure*}
	\vspace{-0.5em}
	\includegraphics[width=\textwidth]{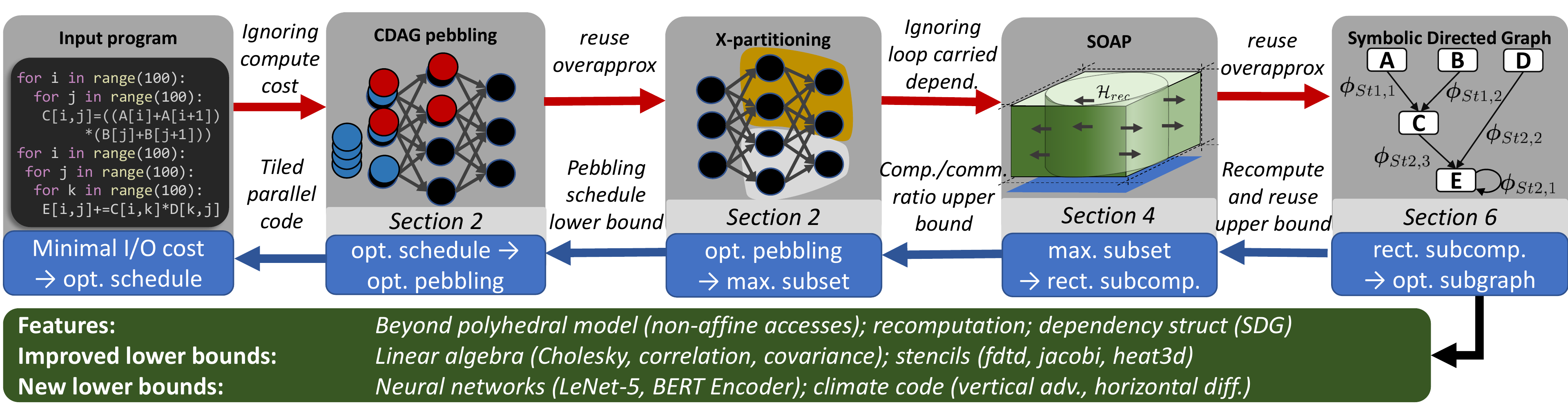}
	\vspace{-2em}
	\caption{High level overview of the combinatorial SOAP 
		analysis. 
		An input program's schedule is modeled as the 
		red-blue pebble game. The 
		\xparting abstraction relaxes the pebbling problem to 
		the graph 
		partition 
		problem. The SOAP abstraction utilizes the static 
		loop structure to 
		upper-bound 
		the size of the optimal \xpart. The Symbolic Directed 
		Graph (SDG) models 
		inter-statement data dependencies. Our method derives 
		I/O lower bounds 
		together with accompanying tile sizes and loop 
		fusions that can be used by 
		a compiler to generate an I/O optimal parallel code.}
	\vspace{-1em}
	\label{fig:soap_overview}
\end{figure*}

In our work, we take a different approach 
based on a combinatorial method. We directly map each elementary computation 
to a vertex in a parametric CDAG, which allows us not only to operate on 
unstructured iteration domains, but also to precisely count the sizes of 
dominator sets and model vertex recomputation. Furthermore, to handle 
complex data dependencies in programs that evaluate multiple arrays, we 
introduce the Symbolic Directed Graph (SDG) abstraction, which encapsulates the 
data flow between elementary computations. This allows us to cover a wider 
class of programs and handle more complex data flow.

To enable precisely mapping every data access to the parametric CDAG vertex, 
we introduce a class of 
\textbf{\textit{Simple Overlap Access Programs}} (SOAP), and present a general 
method to derive \textit{precise} I/O bounds of programs in this class. 
Specifically, SOAPs are defined as loop nests of statements, whose data access 
sets can be modeled as injective functions, and their per-statement data 
overlap can be expressed with constant offsets. 
For programs that do not directly adhere to SOAP, with nontrivial overlaps and 
non-injective access functions, we show that under a set of assumptions, we can 
construct SOAP ``projections'' of those programs, which can be analyzed in the 
same way. Our method strictly contains the polyhedral model and 
associated analysis methods.

To show the breadth of our approach, we demonstrate SOAP analysis on a set of 38 applications, taking Python and C codes as input to create the SDG.
This automated analysis procedure generates symbolic bounds, which match or 
improve upon previously-known results. Notably, we tighten the known I/O lower 
bounds for numerous programs, including stencils by up to a factor of 14, linear algebra kernels (e.g., Cholesky factorization by a factor of two), and the core convolution operation in deep learning by a factor 
of 8. 

Since our derivation of the bounds is 
constructive --- i.e., it provides loop tilings and fusions 
after relaxing loop-carried dependencies --- the results can 
be used by a compiler to generate I/O optimal parallel 
codes. This can both improve existing schedules and possibly reveal new 
parallelization dimensions.

\noindent
The paper makes the following contributions:
\begin{itemize}[leftmargin=*]
	\item A combinatorial method for precisely counting the number of data accesses in parametric CDAGs.
	\item A class of programs --- SOAP --- on which I/O lower bounds can be 
	automatically derived. 
	\item Symbolic dataflow analysis that extends SOAP to 
	multiple- \linebreak statement programs, capturing input and output reuse 
	between 
	statements, as well as data
	recomputation.
	\item I/O analysis of 38 scientific computing kernels, improving existing 
	bounds~\cite{olivry2020automated,chineseCNN1} by up to a factor of 14, and 
	new 
	lower bounds for applications in deep learning, unstructured physics 
	simulation, and numerical weather prediction.
\end{itemize}

	\vspace{-1em}
\section{Background}
\label{sec:background}

We first present several fundamental concepts used throughout the paper. We 
introduce program, memory, and execution models that are based on the work by 
Hong and Kung~\cite{redblue}. We then present a general approach for deriving 
I/O lower bounds based on graph partitioning abstractions. 
The bird's eye view of our method is presented in 
Figure~\ref{fig:soap_overview}.

	\vspace{-0.5em}
\subsection{General Approach of Modeling I/O Costs}
\noindent
\textbf{Program model: CDAG. }
One of the most expressive ways to model executions of arbitrary programs is a
Computation Directed Acyclic Graph (CDAG)~\cite{redblue, redblueHierarchy,
redbluewhite, COSMA, chineseCNN1} $G=(V,E)$, where vertices represent data 
(either
inputs or results of computations) and edges represent data dependencies. That
is, for $u, v \in V$, a directed edge $(u,v) \in E$ signifies that $u$ is
required to compute $v$. Given vertex $v$, vertices $\{u : (u,v) \in E\}$ are
referred to as \emph{parents} of $v$. Analogously, $\{u: (v,u) \in E\}$ are
\emph{children} of $v$. Vertices with in-degree (out-degree) zero are denoted
\emph{program inputs} (\emph{program outputs}).

\noindent
\textbf{Memory model: red-blue pebble game~\cite{redblue}.}
Programs
are executed on a sequential machine equipped with a two-level memory system,
which consists of a fast memory of limited size and unlimited slow memory.  The
contents of the fast memory are represented by $S$ red pebbles. A red pebble
placed on a vertex indicates that the data associated with this vertex resides
in the fast memory. Analogously, data residing in the slow memory is
represented with blue pebbles (of unlimited number).

\noindent
\textbf{Execution model: graph pebbling. } An 
execution of a program represented by a CDAG $G=(V,E)$ is modeled as a sequence
of four allowed pebbling moves: 1) placing a red pebble on a vertex which has a
blue pebble (load), 2) placing a blue pebble on a vertex which has a red pebble
(store), 3) placing a red pebble on a vertex whose parents have red pebbles
(compute) 4) removing any pebble from a vertex (discard). At the program start,
all input vertices have blue pebbles placed on them. Execution finishes when
all output vertices have blue pebbles on them.  A sequence of moves leading
from the start to the end is called a graph \emph{pebbling}~$P$. The number of
load and store moves in $P$ is called the \emph{I/O cost of $P$}.
\textbf{\emph{The I/O cost $Q$ of a program $G$ is the minimum cost among all
valid pebbling configurations}}.  A pebbling with cost $Q$ is called optimal.

\begin{figure*}[h]
	\vspace{-1.5em}
	\includegraphics[width=1.0\textwidth]{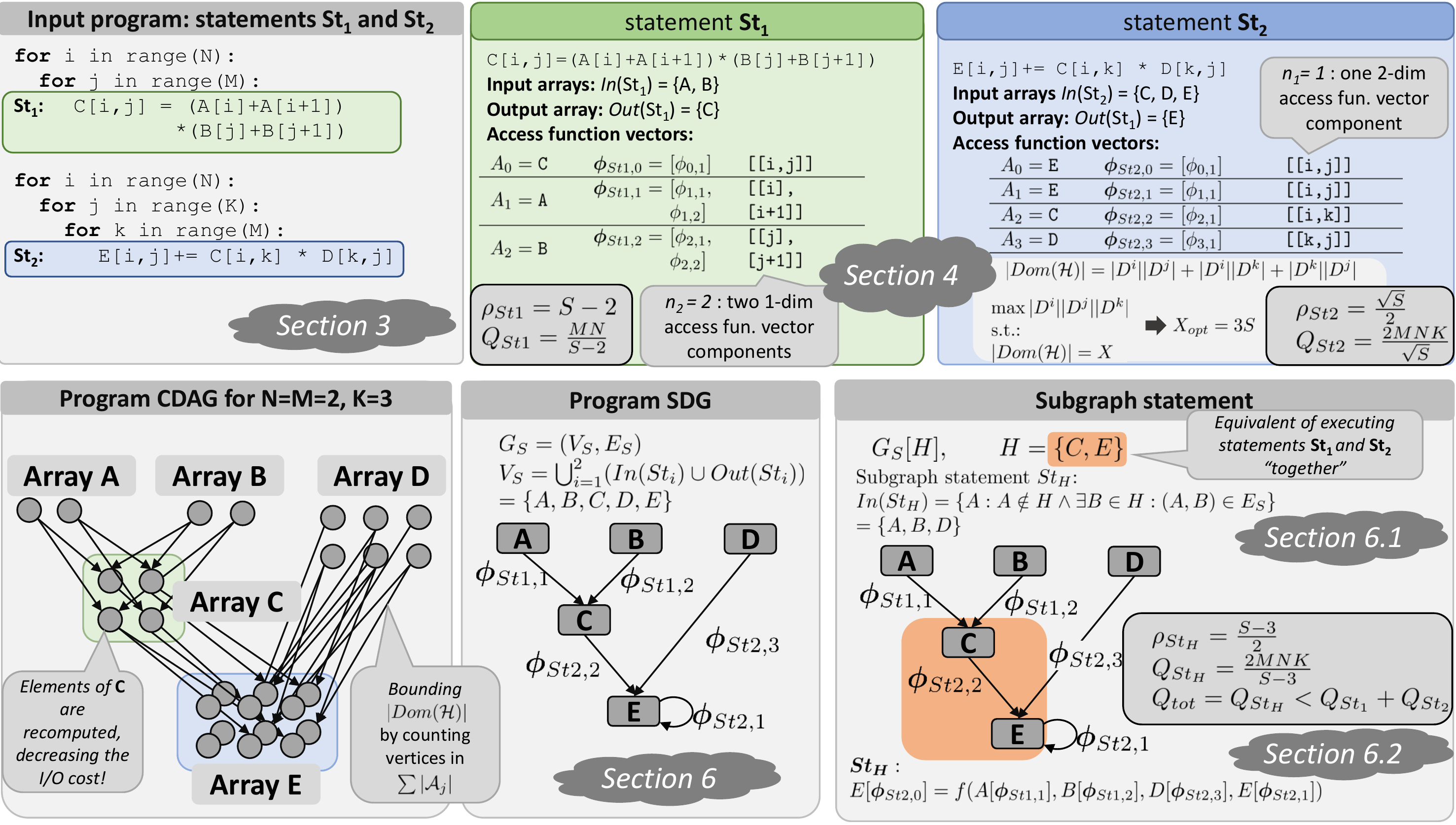}
	\vspace{-1.0em}
	\caption{From the input code to the I/O lower bounds. 
		First, for each 
		statement, the access function vectors $\bm{\phi}$ 
		are extracted from the 
		input program (green and blue fields).  For each 
		statement, the size of its 
		dominator set is obtained using 
		Lemma~\ref{lma:rectTilingAPP} 
		(Section~\ref{sec:maxsubcomp}), and then, the 
		I/O lower bound is obtained using 
		inequality~\ref{eq:final_bound} 
		(Section~\ref{sec:final_bound}). For programs that 
		contain multiple 
		statements, the SDG is constructed 
		(Section~\ref{sec:multistatement}) and 
		all valid subgraph statements are evaluated 
		(Section~\ref{sec:sdg_subgraphs}). Lastly, the final 
		I/O lower bound is 
		obtained (Section~\ref{sec:sdg_lowerbounds}).
	}
	\vspace{-1em}
	\label{fig:soap_flow}
\end{figure*}

\vspace{-1.0em}
\subsection{I/O Lower Bounds}
\label{sec:intro_lowbound}
\noindent
Assume that the optimal pebbling $P_{opt}$ is given. For any constant
$X > S$ we can 
partition this sequence of moves into subsequences, such 
that in each subsequence except of the last one, exactly 
$X-S$ load/store moves are performed (the last subsequence contains at most 
$X-S$ 
load/store moves). Denote the number of these subsequences as $h$. Then observe 
that 
$(X-S)(h-1) \le Q \le (X-S)h$.

\noindent
\textbf{Graph pebbling vs graph partitioning. } Since finding 
$P_{opt}$ is 
PSPACE complete~\cite{redblueHard_}, we seek to 
derive a lower bound of $Q$ from the structure of $G$. Observe that 
the set of vertices which are computed in each subsequence 
defines a 
subgraph $\mathcal{H} \subseteq G$.
By this 
construction, computing vertices in $\mathcal{H}$ requires 
$X-S$ load/store operations in the optimal schedule. The 
number of subsequences $h$ may be bounded by a particular 
partitioning of 
$G$. To do this, we need to introduce two vertex sets defined 
for any 
subgraph of $G$. 

\noindent
\textbf{Dominator and minimum sets~\cite{redblue}.}
Given $\mathcal{H} 
\subseteq G$, a \emph{dominator set} 
$\mathit{Dom}(\mathcal{H})$ is a set of vertices such 
that every path from an input to any vertex in $\mathcal{H}$  
must 
contain at least one vertex in 
$\mathit{Dom}(\mathcal{H})$. 
The \emph{minimum set} 
$\mathit{Min}(\mathcal{H})$ is a
set of all vertices in $\mathcal{H}$ that do not have any 
child in $\mathcal{H}$.
To avoid the ambiguity of non-uniqueness of dominator set 
size,
we denote a \emph{minimum dominator set} 
$\mathit{Dom}_{min}(\mathcal{H})$ to be a dominator set with 
the 
smallest size.

\noindent
\textbf{\xparting: bounding I/O cost. }
Introduced by Kwasniewski et al.~\cite{COSMA}, \xparting 
generalizes the 
\emph{S-partitioning} from Hong and 
Kung~\cite{redblue}.
Given a constant $X$, an \xpart of $G=(V,E)$ 
is a 
collection of $s$ mutually disjoint subsets $\mathcal{H}_i 
\subseteq V$ (referred to as 
\emph{subcomputations})
$\mathcal{P}(X) = \{\mathcal{H}_1, 
\dots, \mathcal{H}_s\} :$ $\forall_{i \ne j} \mathcal{H}_i 
\cap 
\mathcal{H}_j = \emptyset \land$ 
$\bigcup_i 
\mathcal{H}_i = V$  with two additional 
properties:
\begin{itemize}[leftmargin=*]
	\item there are no cyclic dependencies between 
	subcomputations: \linebreak $\forall \mathcal{H}_i \ne 
	\mathcal{H}_j : (\exists 
	(u_1,v_1) \in E \text{ s.t. } u_1 \in \mathcal{H}_i \land v_1 \in 
	\mathcal{H}_j) \implies (\nexists (v_2,u_2) \in E \text{ s.t. } u_2 \in 
	\mathcal{H}_i \land v_2 \in \mathcal{H}_j)$
	\item $\forall \mathcal{H} \in \mathcal{P}(X)$, 
	$\left|{Dom}_{min}\left( 
	\mathcal{H} \right)\right| \le X$ and  $\left| 
	{Min}\left(\mathcal{H}_h\right) \right| \le 
	X$.
\end{itemize}
The authors prove that for any $X>S$, the optimal pebbling 
$P_{opt}$ has an associated \xpart $\mathcal{P}_{opt}(X)$ 
s.t. 
$|\mathcal{P}_{opt}(X)| = h$. 

\noindent
\textbf{Computational intensity.} 
In previous works it was proven that (a) $Q$ is lower bounded by the 
number of 
subsequences $h$ in the optimal pebbling 
$P_{opt}$~\cite{redblue}; (b) $h$ is lower bounded by the size of 
the smallest \xpart $|\mathcal{P}_{min}(X)|$ for any value of 
$X > S$~\cite{COSMA}; (c) 
$|\mathcal{P}_{min}(X)|$ is bounded by the maximum size of a 
single subcomputation $|\mathcal{H}_{X,max}|$ in any valid 
\xpart: $|\mathcal{P}_{min}(X)| \ge 
{|V|}/{|\mathcal{H}_{X,max}|}$ \cite{COSMA}; and (d) if 
$|\mathcal{H}_{X, max}|$ can be expressed as a 
function of $X$, that is, $\chi(X) \equiv |\mathcal{H}_{X, max}|$, 
then 
$Q$ is bounded by
\begin{equation}
\label{eq:general_lower_bound}
Q \ge |V|\frac{X_0 - S}{\chi(X_0)},
\end{equation}
\noindent 
where $X_0 = \argmin_X 
\frac{\chi(X)}{X-S}$ (Lemma 2 in Kwasniewski et  
al.~\cite{confluxArxiv}). 
The expression $\rho = \frac{\chi(X)}{X-S}$ is called the
\emph{computational intensity}.

\section{Simple Overlap Access Programs}
\label{sec:soap}

In Section~\ref{sec:background}, we show how
the I/O cost of a program can be bounded by the maximum size 
of a subcomputation $\mathcal{H}$ in any valid \xpart of 
program CDAG. We now introduce \textbf{Simple Overlap Access 
Programs (SOAP)}: a class of programs for which we 
can derive tight analytic bounds of $|\mathcal{H}|$. 
We leverage the SOAP structure and design an end-to-end 
method for deriving I/O lower bounds of input programs 
(summarized in Figure~\ref{fig:soap_flow}).

\noindent
\textbf{What is SOAP?} Before introducing the formal definition, we start with 
an 
illustrative example, which we use in the following sections.

\begin{eg}
	\label{eg:soap_example}
	Consider the following 3-point stencil code (we use the Python syntax in 
	code listings):
\end{eg}
\vspace{-0.5em}
\begin{lstlisting}[basicstyle=\ttfamily\footnotesize]
for t in range(1,T):
  for i in range(t,N-t):
    A[i,t+1]=(A[i-1,t] + A[i,t] + A[i+1,t])/3 + B[i] 
\end{lstlisting}
\noindent
\emph{This is what we will refer to as a single-statement SOAP. The program 
consists of one statement $St: $ \texttt{A[i,t+1]=(A[i,t] + ... } which is 
placed in two nested loops. All accessed data comes from static, disjoint, 
multi-dimensional arrays (\texttt{A} and \texttt{B}). Furthermore, different 
accesses to the same array (array \texttt{A} is referenced by \texttt{[i,t+1]}, 
\texttt{[i-1,t]}, \texttt{[i,t]},\texttt{[i+1,t]}) are offset by a constant 
stride \texttt{[0,1]}, \texttt{[-1,0]},\texttt{[0,0]}, \linebreak 
\texttt{[1,0]}. We 
denote 
such access pattern as a \textbf{simple overlap} and it is a defining property 
of 
SOAP.}

\noindent
\textbf{Why SOAP? } We use the restriction on the access pattern to precisely 
count the number of vertices in 
$Dom(\mathcal{H})$. If we allow arbitrary overlap of array 
accesses, we need to conservatively assume a maximum 
possible overlap of accessed vertices. This reduces the lower 
bound on $|Dom(\mathcal{H})|$, which, in turn, increases the 
upper bound on $|\mathcal{H}|$, providing less-tight I/O 
lower bound for a program.

\noindent
\emph{This is not a 
	fundamental limitation of our method. However, it allows a
	fully automatic derivation of tight I/O lower bounds for input 
	programs. If the restriction is violated, additional 
	assumptions on the access overlap are needed 
	(Section~\ref{sec:beyondSOAP})}.

\noindent
\textbf{SOAP definition. }
A program is a sequence of statements $St_1, \dots, St_k$. Each such statement 
$St$ is a constant time computable function $f$ 
enclosed 
in a loop nest of the 
following form:  

\vspace{-0.5em}
{
	\footnotesize
	\begin{align}
	\nonumber
	\text{\textbf{for }} \psi^1 \in \mathcal{D}^1:& \\
	\nonumber
	\dots & \\
	\nonumber
	\text{\textbf{for }}&  \psi^\ell \in \mathcal{D}^\ell(\psi^1, \dots, 
	\psi^{\ell-1}): \\
	\nonumber
	&St: A_0[\bm{\phi_0}(\bm{\psi})] \leftarrow 
	f(A_1[\bm{\phi_1}(\bm{\psi})], A_2[\bm{\phi_2}(\bm{\psi})], 
	\dots, 
	A_m[\bm{\phi_m}(\bm{\psi})])
	\end{align}
}
\vspace{-0.5em}

\noindent

	\vspace{-1.5em}
	\noindent
where:
\begin{enumerate}[leftmargin=1.0em]
	
\item The \emph{statement} $St$ is nested in a loop nest 
of depth $\ell$.

\item Each loop in the $t$th level, $t = 1,\dots, \ell$ 
is associated with 
its \emph{iteration variable} $\psi^t$, which 
iterates over 
its domain $\mathcal{D}^t \subset \mathbb{N}$. Domain 
$\mathcal{D}^t$ may 
depend on iteration 
variables from outer loops $\psi^1, \dots,\psi^{t-1}$ 
(denoted as $\mathcal{D}^t(\psi^1, \dots, \psi^{t-1})$).

\item All $\ell$ iteration variables form the 
\emph{iteration vector} 
$\bm{\psi}= [\psi^1, \dots, \psi^\ell]$ and we 
define the \emph{iteration 
	domain $\bm{\mathcal{D}}$} as the set of 
all values the iteration 
vector iterates over during the entire execution of the 
program 
$ \bm{\mathcal{D}} \subset 
\mathbb{N}^\ell$.

\item  The 
dimension of array $A_j$ is denoted as 
$dim(A_j)$.

\item Elements of $A_j$ are 
referenced 
by an \emph{access 
	function vector} $\bm{\phi}_{j}$
which maps $dim(A_j)$ iteration variables $\bm{\psi}_j = 
[\psi_j^1, \dots, \psi_j^{dim(A_j)}]$ to \textbf{a set of 
}$\bm{n_j}$ \textbf{elements} from $A_j$, that is 
$\bm{\phi}_{j}: \mathcal{D}_j^1 \times, 
\dots, \times
\mathcal{D}_j^{dim(A_j)}$ $\rightarrow$ 
$\Big(\mathbb{N}^{dim(A_j)}\Big)^{n_j}$. We then write 
$\bm{\phi}_{j} = [\bm{\phi}_{j,1}, \dots, 
\bm{\phi}_{j,n_j}]$, where $\bm{\phi}_{j,k}:  
\mathcal{D}_j^1 
\times, 
\dots, 
\mathcal{D}_j^{dim(A_j)}$ $\rightarrow$ 
$\mathbb{N}^{dim(A_j)}, k = 
1,\ 
\dots, n_j$. Furthermore, all access function 
components 
$\bm{\phi}_{j,k}(\bm{\psi}_j)$ are injective.

\item All $n_j$ access 
function vector's components
are equal up to a constant translation vector, that 
is, $\forall k=1,\dots, n: \bm{\phi}_{j,k}(\bm{\psi}) = 
\bm{\phi}_{j,1}(\bm{\psi}) + \bm{t}_{k}$, where 
$\bm{t}_{k} = [t_{k}^1, \dots, t_{k}^{dim(A)}] \in 
\mathbb{N}^{dim(A_j)}$. We call 
$\bm{\phi}_j$ the \textbf{simple overlap access}.

\item Arrays $A_1, \dots A_m$ are disjoint. If the output array $A_0$ is also
used as an input, that is, $A_0 \equiv A_j, j \ge 1$, then $\bm{\phi}_0 \cup 
\bm{\phi}_j$ is also the simple overlap access (c.f. 
	Example~\ref{eg:soap_example}).

\item 
Each execution of statement 
$St$ is an evaluation of $f$ for a given value of iteration vector $~\bm{\psi}$.
\end{enumerate}

\noindent
\textbf{Iteration variables and iteration vectors. } Formally, an iteration 
variable $\psi^t$ is an iterator: an object which takes values from its 
iteration domain during the program execution. However, if it is 
clear from the context, we will refer to a particular \emph{value} of 
the iteration variable simply as $\psi^t$ (or a value of iteration vector as
$\bm{\psi}$).

\noindent
\textbf{Vertices as iteration vectors. } Since by definition 
of CDAG, each computation corresponds to a different vertex, 
and by definition of SOAP, every statement execution is associated 
with a single iteration vector $\bm{\psi}$, every non-input 
vertex in $G$ is uniquely associated with an iteration vector $\bm{\psi}$. 
Input vertices are referred to by their access 
function vectors $u = A_j[\bm{\phi}_{j,k}(\bm{\psi})]$.
\textbf{We further define CDAG edges as follows:} for every value of iteration 
vector $\bm{\psi}$, we add an edge from all accessed elements to the vertex 
associated with $\bm{\psi}$, that is: $E = \{(u,v): u = 
A_j[\bm{\phi}_{j,k}(\bm{\psi})], v = \bm{\psi}, 
\bm{\psi} \in \bm{\mathcal{D}} \}$.

\noindent
\textbf{\xparting on SOAP's CDAG. } Recall that our objective 
is to bound the maximum size of any subcomputation 
$|\mathcal{H}|$. 
Given 
pebbling $P$ and an associated \xpart $\mathcal{P}(X)$, 
\emph{every subcomputation $\mathcal{H} \in \mathcal{P}(X)$ 
	is 
	therefore associated with the set of iteration vectors $\bm{\psi}$ 
	of the vertices computed in $\mathcal{H}$.}
In the following section we will derive it 
by counting how many non-input vertices (iteration vectors) 
can $\mathcal{H}$ contain by bounding its dominator set size 
$|Dom(\mathcal{H})|$ - again, by counting vertices 
corresponding to each access 
$A_j[\bm{\phi}_{j,k}(\bm{\psi})]$. 

				\hspace{-1.5em}
\begin{table}[t]
	\vspace{0.0em}
	\setlength{\tabcolsep}{2pt}
	\renewcommand{\arraystretch}{1.2}
	\centering
	\footnotesize
	\sf
		\begin{tabular}{@{}l|ll@{}}
			\toprule
			\multirow{9}{*}{\begin{turn}{90}\textbf{SOAP 
			definition 
						(\S~\ref{sec:soap})}\end{turn}}
			& $A_0$ & \makecell[l]{Output array of statement 
			\emph{St} 
			(may 
			overlap \\
			with input arrays).} \\ 
			& $A_j,$~~~$j = 1,\dots, m$ & Mutually disjoint 
			input 
			arrays of 
			statement \emph{St}. \\ %
			& $\bm{\psi} = [\psi^1, \dots, \psi^\ell]$ 
			& 
			Iteration vector composed of $\ell$ iteration 
			variables.\\
			& $\mathcal{\bm{D}} \subseteq \mathcal{D}^1 
			\times, \dots, \times 
			\mathcal{D}^\ell$
			& 
			\makecell[l]{Iteration domain: a set of 
			values that iteration \\ vector
			$\bm{\psi}$ takes during the entire program 
			execution.\vspace{0.4em}}\\ 
			& $\bm{\phi}_{j} = [\bm{\phi}_{j,1}, \dots, 
			\bm{\phi}_{j,n_j}]$ &  \makecell[l]{Access 
			function vector 
			that maps 
				$dim(A_j)$  
				variables \\
				$[\psi$\scalebox{0.7}{$_j^1$}$, \dots, \psi$ 
				\scalebox{0.7}{$_j^{dim(A_j)}$} $]$
				to $n_j$ elements in array 
				$A_j$.\vspace{0.4em}}\\
			& \scalebox{0.85}{$\bm{t}_{j,k} = [t_{j,k}^1, 
			\dots, 
			t_{j,k}^{dim(A_j)}]$} &
			\makecell[l]{Translation vector of $k$-th access 
			function vector's\\ component $\bm{\phi}_{j,k}$,
			that is $\bm{\phi}_{j,k} \equiv \bm{\phi}_{j,1}, 
			k 
			= 
			1, \dots, n_j$
			}\\
			\midrule
			\multirow{13}{*}{\begin{turn}{90}\textbf{ 
						Single-statement subcomputation 
						(\S~\ref{sec:single_statement})}\end{turn}}
			& $\mathcal{P}(X) = 
			\left\{\mathcal{H}_1,\dots,\mathcal{H}_s\right\}$ 
			&  
					\makecell[l]{An 
						\xpart of CDAG $G=(V,E)$ composed \\ 
						of 
						$s$ 
						disjoint subcomputations.
					}\\		
			& $\bm{D} = D^1 \times, \dots, \times D^\ell$ & 
			\makecell[l]{Subcomputation domain: a Caresian 
			product of \\ ranges 
			of $\ell$ iteration variables 
			during 
			$\mathcal{H}$.\vspace{0.3em}}\\	
			& $\mathcal{H} \subseteq \bm{D} \subseteq V$& 
		\makecell[l]{Subcomputation $\mathcal{H}$ 
			uniquely defined 
			by 
			a set \\ of $|\mathcal{H}|$
			iteration vector's 
			values $\bm{\psi} \in \bm{D}$ taken \\ during 
			$\mathcal{H}$.
			If $\mathcal{H} = \bm{D}$, we call it a 
			\emph{rectangular} \\ \emph{subcomputation} 
			$\mathcal{H}_{rec}$.\vspace{0.3em}}\\
			& $\mathcal{A} = \bm{\phi}[\mathcal{H}]$ & 
			\makecell[l]{Access 
			set: a set of vertices from array $A$ that are \\ 
			accessed by $\bm{\phi}$ during $\mathcal{H}$.} \\
& $\hat{t}^i = \{t_{1}^i, \dots, 
t_{n}^{i}\} \setminus \{0\}$&
\makecell[l]{Access offset set: set of all non-zero $i$th 
coordinates \\
among $n$ translation vectors $\bm{t}_k, k = 1, 
\dots, n$.
}\\	
			& $\mathit{Dom}(\mathcal{H})$ & Dominator set of 
			subcomputation $\mathcal{H}$.\\		
			& $\rho$ & The computational intensity of 
			the \xpart.\\
			& $Q \ge |\bm{\mathcal{D}}|\frac{\sum_{j = 
					1}^{m}|\mathcal{A}_j(X_0)| - 
					S}{\prod_{t=1}^{\ell}|D^t(X_0)|}$ & A 
					number of I/O operations of a schedule. 
			\\			
			\midrule
			\multirow{5}{*}{\begin{turn}{90}
					\textbf{\hspace{1.5em}SDG 
						(\S~\ref{sec:multistatement})}
			\end{turn}} 
			& $G_S = (V_S, E_S)$ &  \makecell[l]{Symbolic 
			Directed Graph, where every array \\ accessed in 
			a 
			program is a vertex, and edges \\ represent data 
			dependencies between them.	\vspace{0.2em}}\\
			& $I \subset V_S$ &  \makecell[l]{Set of 
			read-only arrays of the program.}\\
			& $G_S[H], H \subset V_S \setminus I$ &  
			\makecell[l]{SDG subgraph that represents a 
			subcomputation \\ in which
			 at least one vertex from 
			every \\ array in $H$ is computed.\vspace{0.2em}} 
			\\
			& $\text{\emph{St}}_H$ &  \makecell[l]{Subgraph 
			SOAP 
			statement.}\\		
			\bottomrule
		\end{tabular}%
	\caption{
\vspace{0.5em}{Notation used in the paper.}
	}
	\vspace{-2.5em}
	\label{tab:symbols}
\end{table}

	\vspace{-0.5em}
\section{I/O Lower Bounds For Single-Statement SOAP}
\label{sec:single_statement}
We now derive the I/O bounds for programs 
that contain only one SOAP statement. We start with 
introducing necessary definitions that allow us to bound the 
size of a \emph{rectangular subcomputation}. The summary of 
the notation is presented in Table~\ref{tab:symbols}.
\vspace{-0.5em}
\subsection{Definitions}
\label{sec:single_statement_definitions}

\begin{defn} \textbf{Subcomputation domain.} Denote the 
	set of all values which 
	iteration variable $\psi^t$
	takes during 
	subcomputation 
	$\mathcal{H}$ as $D^t \subset \mathcal{D}^t, t = 
	1,\dots,\ell$. Then, the \textbf{subcomputation domain}  
	$\bm{D}(\mathcal{H}) \subseteq \bm{\mathcal{D}}$ is a 
	Cartesian product of 
	ranges of all $\ell$ iteration variables which they take 
	during $\mathcal{H}$, that is 
	$\bm{D}(\mathcal{H}) = 
	D^1 
	\times, \dots,\times D^\ell$. We therefore have $\mathcal{H} 
	\subseteq \bm{D}(\mathcal{H})
	\subset 
	\mathbb{N}^\ell$. If it is clear from the context, we will 
	sometimes denote $\bm{D}(\mathcal{H})$ simply as $\bm{D}$.
\end{defn}

\begin{eg}
	Recall the program from Example~\ref{eg:soap_example}.
	Consider subcomputation $\mathcal{H}$ in which 
	$\mathtt{t} \in 
	\{1,2\}$ and $\mathtt{i} \in \{1,2\}$. Then, 
	subcomputation domain $\bm{D} = \{1,2\} \times \{1,2\} 
	= 
	\{[1,1],[1,2], $
	$[2,1], [2,2]\}$, but computation itself can contain at 
	most 3 elements $\mathcal{H} \subseteq $
	$\{[1,1], [1,2], [2,2]\}$, since $\bm{\psi} = [2,1] \notin 
	\bm{\mathcal{D}}$ does not belong to the iteration 
	domain. 
\end{eg}

\begin{defn}
	\textbf{Access set and access subdomain.} 
	Consider input array $A$ and its access function vector 
	$\bm{\phi}$.
	Given 	 
	$\mathcal{H}$, the \textbf{access set} 
	$\mathcal{A}$ of $A$ is the set of vertices belonging to  
	$A$ 
	that are accessed during $\mathcal{H}$, that is 
	$\mathcal{A} = \bm{\phi}[\mathcal{H}] = 
	\{A[\bm{\phi}(\bm{\psi})]: \bm{\psi} \in \mathcal{H}\}$. 
	If function $\bm{\phi} = 
	[\bm{\phi}_1, \dots, \bm{\phi}_n]$ accesses $n$ vertices 
	from $A$, we analogously define access sets for each  
	access function component 
	$\bm{\phi}_k[\mathcal{H}], k = 1, \dots, n$. We then have 
	$\mathcal{A} = \bigcup_{k=1}^n\bm{\phi}_k[\mathcal{H}]$.
	The \textbf{access subdomain} $\bm{D}(\mathcal{A})$ is minimum bounding 
	box of the access set $\mathcal{A}$.

\end{defn}

\begin{eg}
For program in Example~\ref{eg:soap_example},
consider subcomputation $\mathcal{H}$ evaluated on only 
one 
iteration vector 
$\mathcal{H} = [i=2,j=2]$.
We have two accessed arrays \texttt{A} and \texttt{B}. Furthermore, we have 
$\bm{\phi}_{\mathtt{A}} = [[i,t+1], [i-1,t], [i,t], [i-1,t]]$. Therefore, 
$dim(\mathtt{A}) = 2$, and 
$\bm{\phi}_{\mathtt{B}}: \mathbb{N}^2 \rightarrow 
\Big(\mathbb{N}^2\Big)^4$. We further have $\bm{\phi}_{\mathtt{B}} = [[i]]$, 
$dim(\mathtt{B}) = 1$, and  $\bm{\phi}_{\mathtt{A}}: \mathbb{N} \rightarrow 
\mathbb{N}$.
\linebreak
To evaluate $St$ for  $\bm{\psi} = [2,2]$, we need to access four   
	elements of \texttt{A} (three loads and one store), so its access set is 
	$\mathcal{A} 
	= \bm{\phi}_{\mathtt{A}}[\mathcal{H}] = \{[2,3], 
	[1,2], [2,2], [2,3]\}$. Furthermore, we have the access 
	subdomain $\bm{D}(\mathcal{A}) $
	$=\{2,3\} \times \{1,2,3\}$.
\end{eg}

\begin{defn}
	\textbf{Access offset set}. Given a simple overlap access
	$\bm{\phi} = [\bm{\phi}_{1}, \dots, \bm{\phi}_{n}]$ 
	consider its $n$ translation vectors $\bm{t}_{k} = 
	[t_{k}^1, \dots, t_{k}^{dim(A_j)}]$
	$ \in \mathbb{N}^{dim(A)}, k = 1, \dots, n$. For each 
	dimension 
	$i = 1,\dots, dim(A_j)$ we denote 
	$\hat{t}^i = \{t_{1}^i, \dots, t_{n}^i\} \setminus \{0\}$ 
	as the set of all unique non-zero $i$th coordinates 
	among all $n$ translation vectors.
\end{defn}

\begin{defn}
	\textbf{Rectangular subcomputation} For a given 
	subcomputation domain $\bm{D}$, a subcomputation 
	$\mathcal{H}$ is called \textbf{rectangular} if $\mathcal{H} 
	= 
	\bm{D}$ and 
	is denoted $\mathcal{H}_{rec}(\bm{D})$. The size of 
	rectangular 
	computation is $|\mathcal{H}_{rec}(\bm{D})| = 
	\prod_{t=1}^{\ell}|D^t|$.
	If it is clear from the context, we will denote 
	$\mathcal{H}_{rec}(\bm{D})$ simply as $\mathcal{H}_{rec}$.
\end{defn}

\begin{observation}
	\label{cor:simple_overlap}
	Consider a simple overlap access $\bm{\phi} = 
	[\bm{\phi}_1, \dots, \bm{\phi}_n]$ of array $A$ and a 
	rectangular subcomputation 
	$\mathcal{H}_{rec}(\bm{D})$.
	Then since all $\bm{\phi}_k$ are equal up to translation, 
	the 
	ranges of iteration variables they access are also equal 
	up 
	to the same translation: $\forall i = 1, \dots, dim(A) : 
	\forall j = 1, \dots, n: 
	\bm{\phi}_j[D^i] = \bm{\phi}_1[D^i] + t_j$, which 
	also implies that 
	$\forall i = 1, \dots, dim(A) : \forall j = 1, \dots, n: 
	|\bm{\phi}_j[D^i]| = |\bm{\phi}_1[D^i]|$.
\end{observation}

To bound the sizes of rectangular subcomputations, we use two lemmas given by 
Kwasniewski et al.~\cite{confluxArxiv}:

\begin{lma}
	\label{lma:vi_bound}
	(Lemma 4 in~\cite{confluxArxiv})
	For statement $St$, given $\bm{D}$, the size of 
	subcomputation $\mathcal{H}$ (number 
	of 
	vertices of $S$ computed during 
	$\mathcal{H}$) is bounded by the sizes of the iteration 
	variables' 
	sets $D^t, t = 
	1, \dots, \ell$:
			\vspace{-0.5em}
	\begin{equation}
	\label{eq:vmax_vol}
	|\mathcal{H}| \le \prod_{t=1}^{\ell}|D^t|.
	\end{equation}
			\vspace{-0.5em}
\end{lma}

\begin{proof}
	Inequality~\ref{eq:vmax_vol} follows from a combinatorial 
	argument: 
	each 
	computation in $\mathcal{H}$ is uniquely defined by its 
	iteration 
	vector $[\psi^1, 
	\dots, 
	\psi^\ell]$. As each iteration variable $\psi^t$ 
	takes $|D^t|$ 
	different values during $\mathcal{H}$, we have $|D^1|
	\cdot 
	\dots \cdot |D^t| 
	= \prod_{t=1}^{\ell}|D^t|$ ways how to uniquely choose 
	the iteration 
	vector in 
	$\mathcal{H}$. 
\end{proof}

\begin{lma}
	\label{lma:projection_bound}
		(Lemma 5 in~\cite{confluxArxiv})
	For the given access function $\bm{\phi} = [\bm{\phi}_1, 
	\dots, \bm{\phi}_n]$ accessing array $A$, 
	$A[\bm{\phi}(\bm{\psi})]$,
	the access set size of each of components 
	$|\bm{\phi}_k[\mathcal{H}]|$ during subcomputation 
	$\mathcal{H}$  is 
	bounded by the 
	sizes of $dim(\bm{\phi}_A)$
	iteration variables' sets $D^i, k = 
	1, \dots, dim(\bm{\phi}_j)$:
	
	\vspace{-0.5em}
	\begin{equation}
	\label{eq:a_j_vol}
	|\bm{\phi}_k[\mathcal{H}]| \le 
	\prod_{i=1}^{dim(\bm{\phi}_A)}|D^i|
	\end{equation}
			\vspace{-0.5em}
	
	where $D^i \ni \psi^i$ is the iteration domain of 
	variable $\psi^i	$ during $\mathcal{H}$.
\end{lma}

\begin{proof}
	We use the same 
	combinatorial argument as in Lemma~\ref{lma:vi_bound}. 
	Since access functions are injective,
	each vertex in $A$ accessed by $\bm{\phi_k}$ is uniquely 
	defined by $[\psi_k^1, 
	\dots, 
	\psi_k^{dim(A)}]$. Knowing the number of different 
	values each 
	$\psi_k^j$ 
	takes in $\mathcal{H}$, 
	we bound the number of different access vectors 
	$\bm{\phi}_k[\mathcal{H}]$.
\end{proof}

		\vspace{-0.5em}
\subsection{Bounding SOAP Access Size}
\label{sec:maxsubcomp}

Recall that our goal is to find the maximum size of the subcomputation given 
its dominator size. We 
first do the
converse: given the rectangular subcomputation $\mathcal{H}_{rec}$, we bound 
the 
minimum number of input vertices required to compute $\mathcal{H}_{rec}$. 
In Section~\ref{sec:compintensity} we prove that indeed $\mathcal{H}_{rec}$ is 
the subcomputation that upper-bounds the maximum 
computational intensity 
$\rho$.
Since arrays $A_1, \dots , A_m$ are disjoint, the total number of 
input vertices 
is the sum of their access set sizes: 
$|Dom_{min}(\mathcal{H}_{rec}) \ge \sum_{j=1}^{m}|\mathcal{A}_j|$. We now 
proceed to bound individual access set sizes $|\mathcal{A}_j|$.

Consider array $A$ with $dim(A) = d$ and its access function 
$\bm{\phi}(\bm{\psi}) = 
[\bm{\phi}_1(\bm{\psi}), \dots,$ $ 
\bm{\phi}_n(\bm{\psi})]$ that access $n$ elements from $A$
(to simplify the notation, we drop the subscript $j$, since we consider only 
one array).
Observe that during $\mathcal{H}_{rec}$, all combinations of iteration 
variables 
$\psi^1 \in D^1, \dots, \psi^\ell \in D^\ell$ are accessed, so 
$|\mathcal{H}_{rec}| = \prod_{t=1}^{\ell}|D^t|$ 
(Lemma~\ref{lma:vi_bound}). This also implies that 
each of $k = 1, \dots, n$ accesses to $A$ required
$|\bm{\phi}_k[\mathcal{H}_{rec}]| 
=\prod_{t=1}^{d}|D^t|$ 
vertices from $A$ (Lemma~\ref{lma:projection_bound} and 
Observation~\ref{cor:simple_overlap}). 
Therefore, the total number of accesses to array $A$ 
during $\mathcal{H}_{rec}$ is 
$|\mathcal{A}| 
\ge \prod_{t=1}^{d}|D^t|$.
However, the sets of vertices accessed by different 
$\bm{\phi}_k$ may overlap, that is, there may exist two 
accesses $\bm{\phi}_l$ and $\bm{\phi}_m$, for which 
$\bm{\phi}_l[\mathcal{H}_{rec}] \cap 
\bm{\phi}_m[\mathcal{H}_{rec}] 
\ne \emptyset$. Therefore, we also obtain the upper bound
$|\mathcal{A}| \le 
\sum_{j=1}^{n}\prod_{t=1}^{d}|D^t|$. We now want to 
narrow the gap between the upper and the lower bounds.

\begin{lma}
	\label{lma:rectTilingAPP}
	If a given input array $A$ with $dim(A) = d$ is accessed 
	by a simple overlap access 
	$\bm{\phi}(\bm{\psi}) = [\bm{\phi}_1(\bm{\psi}), \dots, 
	\bm{\phi}_n(\bm{\psi})]$,
	its access set size $|\mathcal{A}|$
	during rectangular computation $\mathcal{H}_{rec}(\bm{D})$ is 
	bounded by 
	
	\vspace{-0.5em}
	\begin{equation}
	\label{eq:aapAccessSizes}
	|\mathcal{A}| = |\bm{\phi}[\mathcal{H}_{rec}(\bm{D})]|  
	\ge 
	2\prod_{i=1}^{d}|D^i| -
	\prod_{i=1}^{d}(|D^i| - |\hat{t}^i|),
	\end{equation}
	where $|\hat{t}^i|$ is the size of the access offsets set in the $i$th  
	dimension.
\end{lma}

	\vspace{-0.5em}
\begin{proof}
	W.l.o.g., consider the first access function component 
	$\bm{\phi}_1$ 
	and its $\prod_{t=1}^{d}|D^t|$ accessed 
	vertices $\bm{\phi}_1[\mathcal{H}_{rec}]$. We will lower bound the number 
	of 
	accesses to $A$ from remaining $\bm{\phi}_k, k = 2, 
	\dots, n$,  which \emph{do not} overlap with 
	$\bm{\phi}_1[\mathcal{H}_{rec}]$, that is
	$|\bigcup_{k=2}^n\bm{\phi}_k[\mathcal{H}_{rec}] \setminus 
	\bm{\phi}_1[\mathcal{H}_{rec}]|$.
		Since by construction of $\mathcal{H}_{rec}$, all 
	$\bm{\phi}_k[\mathcal{H}_{rec}]$ are Cartesian 
	products of iteration variables' ranges 
	$\bm{\phi}_k[D^1] \times, \cdots, $  $\times 
	\bm{\phi}_k[D^d]$, there is a 
	bijection between $\bm{\phi}_k[\mathcal{H}_{rec}]$ and an 
	$d$-dimensional hyperrectangle $H_k \in \mathbb{N}^d$. 
	To secure correctness of our lower bound on 
	$|\mathcal{A}|$, we need to find the volume of the 
	smallest union of these hyperrectangles.
	
 	Note that $|\hat{t}^i|$ is a lower bound on the maximum offset between 
 	any two $H_j \ne H_k$  in dimension $i$: the union of all hyperrectangles 
 	$\bigcup_{k=1}^n H_k$
 	``stretches'' at least $|D^i| + |\hat{t}^i|$ elements in the $i$th 
 	dimension 
 	for all $i = 1, \dots, d$ 
 	(see Figure~\ref{fig:access_bound}).
 	To see this, observe that since $D^i \subset N$, for each element in 
	the access offset set $t_j^i \in \hat{t}^i$ 
	there is at least one element in $D^i + t_{j}^i$ 
	that is not in $D^i$, which implies that $|(D^i + 
	t_{j}^i) \setminus D^i| \ge 1$.
	Since $D^i$ is finite, there is a single well-defined maximum and 
		a minimum element, which implies that $(\max\{D^i\} + t_{j}^i 
		\notin D^i) \lor (\min\{D^i\} + t_{j}^i \notin D^i)$. 
	Also, because by definition of $\hat{t}^i$ we have 
	$\forall t_j^i, t_k^i \in \hat{t}^i: t_j^i \ne t_k^i$, 
	then we also have that each $t_j^i$ accesses at least one 
	``non-overlapping'' element independent of any other 
	$t_k^i$, that is $\forall t_j^i, t_k^i \in \hat{t}^i: 
	\max\{D^i\} + t_j^i \ne \max\{D^i\} + t_k^i$.	
	
	\begin{figure}
			\vspace{-0em}
		\includegraphics[width=\columnwidth]{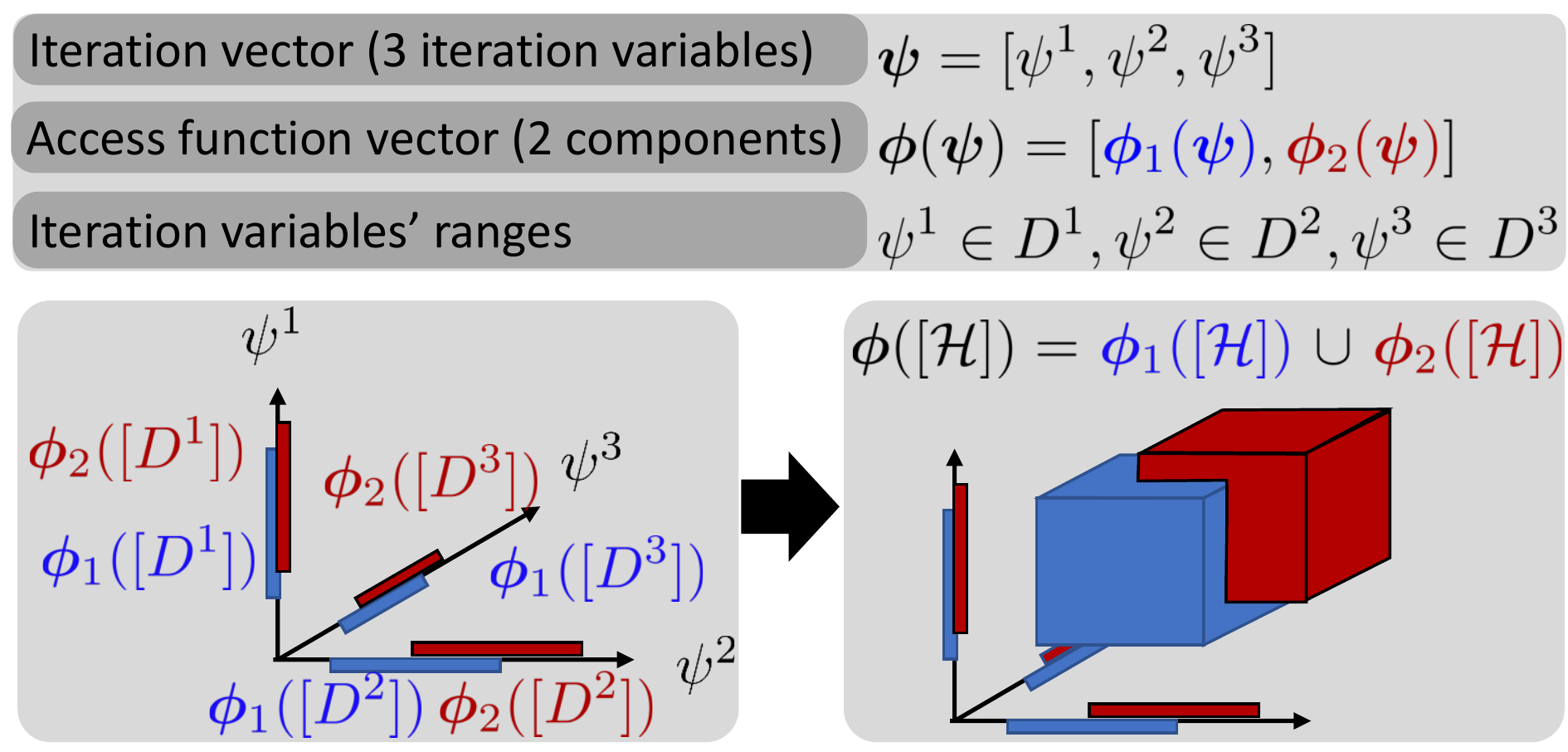}
			\vspace{-1em}
		\caption{Intuition behind Lemma~\ref{lma:rectTilingAPP}. Access sets 
		$\bm{\phi}_[\mathcal{H}_{rec}(\bm{D})]$ as 
			3-dimensional hyperrectangles. The union 
			$|\bigcup_{k=1}^n\bm{\phi}_k[\mathcal{H}_{rec}]|$ (and 
			therefore, the 
			total number of accesses 
			$|\mathcal{A}|$ ) is minimized when the hyperrectangles 
			are 
			placed in 
			two antipodal locations of the subcomputation domain
			$\mathcal{D}$.}
			\vspace{-1em}
		\label{fig:access_bound}
	\end{figure}
	
The arrangement of hyperrectangles $H_k, k = 1, \dots, n$ in a
$\mathbb{N}^d$ lattice s.t., their bounding box is  $\bm{D} = (|D^1| + 
|\hat{t}^1|) \times \dots \times (|D^d| + |\hat{t}^d|)$, 
which minimizes the size of their union $|\bigcup_k H_k|$ satisfies two 
properties:
\begin{enumerate}
	\item  there exist two ``extreme'' $H_p$, $H_q$, such that $H_q = H_p + 
	\bm{v}$, $\bm{u} = \mathbb{Z}^d, \forall_{i = 1, \dots, d}: |v^i| = 
	|\hat{t}^i|$,
\item all the remaining $H_k, k \ne p, q$ perfectly overlap with the 
``extreme'' 
hyperrectangles $H_k \subseteq H_p \cup H_q$.
\end{enumerate}

To see this, observe that for every non-zero $|\hat{t}^i|$ we need two 
hyperrectangles $H_p^i \ne H_q^i$, s.t., $H_q^i = H_p^i + [\cdot, \dots, 
|\hat{t}^i|, 
\dots, \cdot]$, that is, $H_q^i$ is offset from $H_p^i$ by  $|\hat{t}^i|$ in 
$i$th 
dimension. We therefore have $\bigcup_{i, |\hat{t}^i| > 0} (H_p^i \cup H_q^i) 
\subseteq   \bigcup_k H_k$. Since $H_p^i$ and $H_q^i$ are pairwise non-equal, 
but there are no restrictions on $H_p^i, H_p^j, i \ne j$, we have that the union
$\bigcup_{i, |\hat{t}^i| > 0} (H_p^i \cup H_q^i)$ is minimized if $\forall_{i 
\ne j} H_p^i = H_p^j$.

	Finally, observe the volume of $|\bigcup_{k=1}^n H_k|$ s.t. to the claimed 
	arrangement is:

	\vspace{-0.5em}
	\begin{equation}
	\label{eq:union_bound}
	\Big|\bigcup_{k=1}^n H_k\Big| = |H_p \cup H_q| =
	2\prod_{i=1}^{d}|D^i| -
	\prod_{i=1}^{d}(|D^i| - |\hat{t}^i|)
	\end{equation}
	\vspace{-0.5em}

	It shows that for any set of $n$ hyperrectangles 
	s.t. the given constraint, the volume of their union is no 
	smaller than the one in Equation~\ref{eq:union_bound}. 
	Since the offset constraint is also a lower bound on the 
	number of non-overlapping accesses in each dimension, it 
	also forms the bound on 
	$|\bigcup_{k=1}^n\bm{\phi}_k[\mathcal{H}_{rec}]| = 
	|\bm{\phi}[\mathcal{H}_{rec}] |= |\mathcal{A}|$.
\end{proof}

\vspace{-1.0em}
\subsection{Input-Output Simple Overlap}
\label{sec:input-output}
If one of the input arrays $A_i, i 
\ge 
1$, is also the output array $A_0$, then their access function vectors 
$\bm{\phi}_0$ and $\bm{\phi}_i$ form together a simple overlap access 
(Section~\ref{sec:soap}).
In such cases, some vertices 
accessed by 
$\bm{\phi}_i$ during $\mathcal{H}_{rec}$ may be computed and do not 
need to be loaded. We formalize it in the following corollary, which follows 
directly from Lemma~\ref{lma:rectTilingAPP}:

\begin{crl}
\label{crl:inputOutput}
Consider statement $St$ that computes array $A, dim(A) = d$ 
and 
simultaneously 
accesses it as an input $A[\bm{\phi}_0(\bm{\psi})] = 
f(A[\bm{\phi}_1(\bm{\psi})])$. If $\bm{\phi}_0 \cup 
\bm{\phi}_1$ is a simple overlap access,
the access set size $|\mathcal{A}|$
during rectangular computation $\mathcal{H}_{rec}$ is 
bounded by 
		\vspace{-0.5em}
\begin{equation}
\label{eq:aapAccessSizesWithOutput}
|\mathcal{A}|
\ge 
\prod_{i=1}^{d}|D^i| -
\prod_{i=1}^{d}(|D^i| - |\hat{t}^i|),
\end{equation}
		\vspace{-0.5em}
where $\hat{t}$ is an access offset offset set of $\bm{\phi}_0 \cup 
\bm{\phi}_1$.
\end{crl}
\begin{proof}	
	This result follows directly from Lemma~\ref{lma:rectTilingAPP}. Since 
	there are at least $2\prod_{i=1}^{d}|D^i| -
	\prod_{i=1}^{d}(|D^i| - |\hat{t}^i|)$ vertices accessed from $A_i$, and at 
	most $\prod_{i=1}^{d}|D^i|$ of them can be computed during 
	$\mathcal{H}_{rec}$ (Lemma~\ref{lma:projection_bound}) and therefore, do 
	not have to be loaded, then at least $2\prod_{i=1}^{d}|D^i| -
\prod_{i=1}^{d}(|D^i| - |\hat{t}^i|) - \prod_{i=1}^{d}|D^i|$ elements have to 
be accessed from the outside of $\mathcal{H}_{rec}$.
\end{proof}

\subsection{Bounding Maximal Subcomputation}
\label{sec:compintensity}

In Section~\ref{sec:maxsubcomp} we lower-bounded the dominator set size of the 
rectangular subcomputation $|Dom_{min}(\mathcal{H}_{rec})| = 
\sum_{j=1}^{m}|\mathcal{A}_j|$ by bounding the sizes of simple overlap access 
sets sizes $|\mathcal{A}_j|$ (Lemma~\ref{lma:rectTilingAPP}).
Recall that to bound the I/O lower bound we need the size $\chi(X)$ of the 
\emph{maximal} subcomputation $\mathcal{H}_{max}$ for given value of $X$ 
(Inequality~\ref{eq:general_lower_bound}). We now prove that 
$\mathcal{H}_{rec}$ upper-bounds the size of $\mathcal{H}_{max}$.

Given $\mathcal{H}$, denote the the ratio of the size of the subcomputation to 
the dominator set size $\delta(\mathcal{H}) = 
\frac{|\mathcal{H}|}{\sum_{j=1}^{m}|\bm{\phi}_j[\mathcal{H}]|}$. By definition, 
$\mathcal{H}_{max}$ maximizes $\delta$ among all valid $\mathcal{H} \in 
\mathcal{P}$.
 We need to 
show that for a fixed subcomputation domain $\bm{D}_0$, among all 
subcomputations for which $\bm{D}(\mathcal{H}) = \bm{D}_0$, the rectangular 
subcomputation $\mathcal{H}_{rec}(\bm{D}_0)$ upper-bounds $\delta$. 
\emph{Note that an \xpart derived from the optimal pebbling schedule $P_{opt}$ 
may not include $\mathcal{H}_{rec}$}. However, $\forall X: \chi_{rec}(X) \ge 
\chi(X)$, that is, given $X$, the size of $\mathcal{H}_{rec}$ s.t., 
$\sum_{j=1}^{m}|\bm{\phi}_j[\mathcal{H}_{rec}]| = X$ will always be no smaller 
than the size of $\mathcal{H}_{max}$. To show this, we first need to introduce 
some auxiliary definitions.

\textbf{Iteration variables, their indices, and their values. }
To simplify the notation, throughout the paper we used the iteration variables 
$\psi^i$ and 
the \emph{values} they take for some iteration interchangeably. However, now we 
need to make this distinction explicit.
The iteration vector consists of $\ell$ iteration variables $\bm{\psi} = 
[\psi^1, 
\dots,$  $\psi^\ell]$. Each access function $\bm{\phi}_j$ is 
defined on $dim(A_j)$ $\le \ell$ of them. Recall that $\bm{\psi}_j$ is the set 
of 
iteration variables accessed by $\bm{\phi}_j$ (Section~\ref{sec:soap}, property 
(5)). 
To keep track of the indices of particular iteration variables, denote 
$\bm{\Psi} = [\ell] = 
\{1, \dots, \ell\} \subset \mathbb{N}$,  $\bm{\Psi}_j \subseteq \bm{\Psi}$, and 
$\bm{\Psi}'_j = 
\bm{\Psi} \setminus \bm{\Psi}_j$ as the sets of integers. If $i \in 
\bm{\Psi}_j$, then the 
$i$th iteration variable $\psi^i$ is accessed by the access function 
$\bm{\phi}_j$. We further define $\bm{\psi}^* \in \mathbb{N}^\ell$ as a 
specific 
\emph{value} of the iteration vector $\bm{\psi}$ that 
uniquely defines a 
single non-input vertex. We analogously define 
$\bm{\psi}_j^*$, $\psi^{i,*}$, and $\psi_j^{i,*}$ (the last one being a value 
of 
$i$th iteration variable of the $j$th access).
We also define $\theta(\bm{\psi}_j^*, \mathcal{H})$ as the number of vertices 
in $\mathcal{H}$ that have all their $\bm{\Psi}_j$ 
coordinates equal to  
$\bm{\psi}_j^*$, that is $\theta(\bm{\psi}_j^*, \mathcal{H}) = |\{\bm{\psi}^* : 
\bm{\psi}^* \in \mathcal{H} \land (\forall i \in \bm{\Psi}_j: \psi^{i,*} = 
\psi_j^{i,*} )\}|$.

We now formalize our claim in the following lemma:

\begin{lma}
	\label{lma:rhoAPP}
	Given the subcomputation domain $\bm{D}_0$, 
	$\mathcal{H}_{rec}(\bm{D}_0)$ maximizes $\delta(\mathcal{H})$ for all 
	$\mathcal{H}$ s.t. $\bm{D}(\mathcal{H}) = \bm{D}_0$.
	\begin{equation}
	\label{eq:rhoAPP}
	\forall \mathcal{H}:  
	\delta(\mathcal{H}) \le 	\delta(\mathcal{H}_{rec})
	\end{equation}
\end{lma}

\begin{proof}
Instead of maximizing $\delta(\mathcal{H})$, we will minimize 
$\delta^{-1}(\mathcal{H}) = 
(\sum_{j=1}^{m}|\bm{\phi}_j[\mathcal{H}]|)/|\mathcal{H}| = 
\sum_{j=1}^{m}|\bm{\phi}_j[\mathcal{H}]|/|\mathcal{H}|$ over all 
possible $\mathcal{H}$. Observe 
that 
$\delta^{-1}(\mathcal{H})$ is linear w.r.t. the ratios of individual access 
function sets sizes $|\bm{\phi}_j[\mathcal{H}]|$ and the size of 
subcomputation 
$|\mathcal{H}|$. Therefore, we can examine each access 
$\bm{\phi}_j[\mathcal{H}]$ 
separately and show that every 
$\delta^{-1}_j = |\bm{\phi}_j[\mathcal{H}]|/|\mathcal{H}|$ 
is minimized for $\mathcal{H} = \mathcal{H}_{rec}$.	Then, if 
$\mathcal{H}_{rec}$ minimizes each of $\delta^{-1}_j$, then $\delta^{-1} = 
\sum_{j=1}^m \delta^{-1}_j$ is minimized, so indeed $\mathcal{H}_{rec}$ 
maximizes the ratio of the subcomputation size to the dominator set size.

Observe now, that for any $\mathcal{H}$ we have that $\forall_j: 
\delta^{-1}_j$ 
is monotonically decreasing w.r.t. $\theta(\bm{\psi}_j^*, \mathcal{H})$ for all 
$\bm{\psi}_j^* \in \bm{\phi}_j[\mathcal{H}]$. That is - pick any input vertex 
$\bm{\psi}_j^*$ from the set of vertices accessed by 
$\bm{\phi}_j[\mathcal{H}]$. Adding compute vertices $\bm{\psi}^*$ to 
$\mathcal{H}$ that access $\bm{\psi}_j^*$ do not increase the 
access set size  
$\bm{\phi}_j[\mathcal{H}]$, since $\bm{\psi}_j^*$ is already accessed. However, 
it increases the size of $\mathcal{H}$. Clearly, 
$\delta^{-1}_j$ reaches its minimum if $\forall 
\bm{\psi}_j^* \in \bm{\phi}_j[\mathcal{H}] :  \theta(\bm{\psi}_j^*, 
\mathcal{H}) = \prod_{i \in \bm{\Psi}'_j}|D^i|$, that is, $\mathcal{H}$ 
computes all vertices spanned by the access set $\bm{\phi}_j[\mathcal{H}]$ and 
all elements in the Cartesian product of ``free'' (independent of the access 
function $\phi_j$) iteration domains $D^i, i 
\in \bm{\Psi}'_j$.

We showed that for all $j$, given its initial access set 
$\bm{\phi}_j[\mathcal{H}]$, the ratio 
$\delta^{-1}_j$ is minimized for the
``almost-rectangular'' subcomputation, that is, $\mathcal{H}$ which computes 
all vertices $\bm{\psi}^* \in \bm{\phi}_j[\mathcal{H}] \times \prod_{i \in 
\bm{\Psi}'_j}D^i$. 
We now need to show that also extending $\mathcal{H}$ over the ``dependent'' 
ranges $\bm{\Psi}_j$ won't increase the ratio $\delta^{-1}$. When the 
access set size $\bm{\phi}_j[\mathcal{H}]$ increases by a factor $x$, 
$\mathcal{H}$ increases proportionally by $x$ too, keeping the 
ratio constant (See Figure~\ref{fig:rect_subcomp} for an example for 
$\ell = 3$).

\begin{figure}
	\includegraphics[width=1\columnwidth]{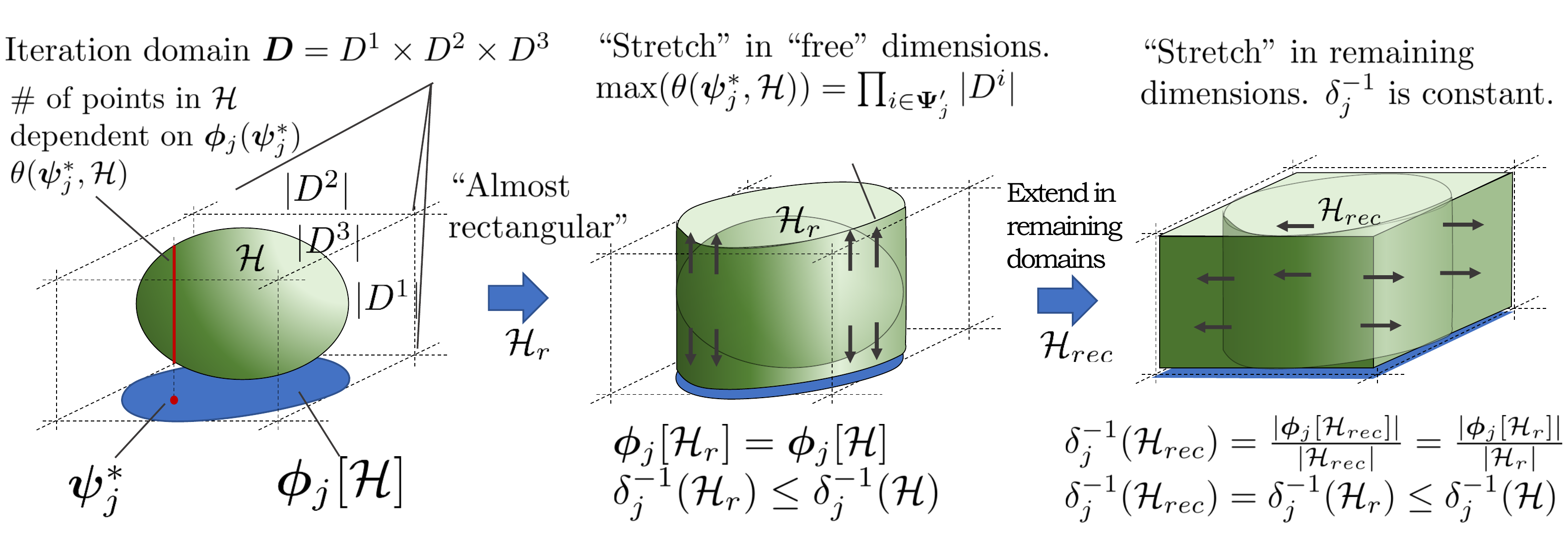}
	\vspace{-1em}
	\caption{Intuition behind Lemma~\ref{lma:rhoAPP}: extending the 
	subcomputation in the free dimensions 
	w.r.t $\bm{\phi}_j$ does not increase 
	$|\bm{\phi}_j[\mathcal{H}]|$. Once the subcomputation is almost 
	rectangular, extending $H$ in the remaining dimensions keeps the 
	ratio $\delta_j^{-1}$ constant.}
\vspace{-1em}
	\label{fig:rect_subcomp}
\end{figure}

Since our goal is to minimize each 
$\delta^{-1}_j$ \emph{separately}, independently of other $\delta^{-1}_i, i\ne 
j$, assume 
that we have already extended $\mathcal{H}$ to the ``almost-rectangular'' 
subcomputation, that is, all combinations of $\prod_{i \in \bm{\Psi}'_j}D^i$ 
were accessed in $\mathcal{H}$.
Observe now that $\theta(\bm{\psi}_j^*, \mathcal{H}) = \prod_{i \in 
\bm{\Psi}'_j}|D^i|$ for \emph{any} vertex $\bm{\psi}_j^*$. Therefore, since 
$|\mathcal{H}| = \sum_{\bm{\psi}_j^* \in \bm{\phi}_j[\mathcal{H}]} \prod_{i \in 
\bm{\Psi}'_j}|D^i|$, we see that $\delta^{-1}_j$ is \emph{constant} w.r.t., the 
size of the access set: $\delta^{-1}_j = 
\frac{|\bm{\phi}_j[\mathcal{H}]|}{|\mathcal{H}|} = \frac{1}{\prod_{i \in 
\bm{\Psi}'_j}|D^i|}$. Therefore, we can safely maximize 
$\bm{\phi}_j[\mathcal{H}]$ to the entire access set of the rectangular 
subcomputation $\mathcal{H}_{rec}$ without increasing $\delta_j^{-1}$. We 
conclude that for every access function 
$\bm{\phi}_j$ and every iteration variable index $i$, evaluating all vertices 
$\bm{\psi}*$ s.t. $\psi^i$ iterates over the entire domain $D^i$ minimizes 
$\delta_j^{-1}$.

\end{proof}
	
\subsection{I/O Lower Bounds and Optimal Tiling}
\label{sec:final_bound}
We now proceed to the final step of finding the I/O lower bound. Recall from 
Section~\ref{sec:intro_lowbound}, that the last missing piece is $\chi(X)$; 
that is, we seek to express $|\mathcal{H}_{max}(\bm{D})| = 
\prod_{t=1}^\ell|D^t|$ 
as a function of $X$. Observe that by Lemma~\ref{lma:rhoAPP}, 
$|Dom_{min}(\mathcal{H}_{max}(\bm{D})) \ge \sum_{j = 
	1}^{m}(2\prod_{i=1}^{dim(A_j)}|D_j^i| -
\prod_{i=1}^{dim(A_j)}(|D_j^i| - |\hat{t}_j^i|)$. On the other hand, 
by definition of \linebreak \xparting, 
$|Dom_{min}(\mathcal{H}_{max}(\bm{D}))| \le X$. 
Combining these inequalities, we solve for all $|D^t|$ as functions of $X$ by 
formulating it as 
the 
	optimization problem (see Section 3.2 
	in Kwasniewski et al.~\cite{confluxArxiv}):	
	
	\vspace{-1.0em}
\begin{align}
\label{eq:findingX}
\nonumber
\max \prod_{t=1}^{\ell}|D^t| & \hspace{2em}\text{s.t.}\\
\nonumber
\sum_{j = 
	1}^{m}|\mathcal{A}_j| \le X \\
\forall 1 \ge t \ge \ell : |D^t|& \ge 1 
\end{align} 
	\vspace{-0.5em}

Solving the above optimization problem yields $\chi(X) = 
|\mathcal{H}_{max}(\bm{D})|$. Since 
Lemma~\ref{lma:rhoAPP} gives a valid upper bound on computational intensity 
for \emph{any} value of $X$, we seek to find the tightest (lowest) upper bound. 
One can obtain $X_0 = \argmin_X \frac{\chi(X)}{X-S}$, since $\chi(X)$ 
is 
differentiable. Finally, combining Lemma~\ref{lma:rectTilingAPP}, 
inequality~\ref{eq:general_lower_bound}, and the optimization 
problem~\ref{eq:findingX}, we obtain the I/O lower bound for the 
single-statement SOAP program:

\begin{equation}
\label{eq:final_bound}
Q \ge |\bm{\mathcal{D}}|\frac{\sum_{j = 
		1}^{m}|\mathcal{A}_j(X_0)| - S}{\prod_{t=1}^{\ell}|D^t(X_0)|},
\end{equation}
\noindent
where $|\mathcal{A}_j(X_0)|$ are the access set sizes obtained from 
Lemma~\ref{lma:rectTilingAPP} for the optimal value of $|D^i|$ derived from the 
optimization problem~\ref{eq:findingX}.
	
Substituting $X_0$ back to $|D^t|(X)$ has a direct interpretation: they 
constitute optimal loop tilings for the maximal subcomputation. Note that such 
tiling might be invalid due to problem relaxations: e.g., we ignore 
loop-carried dependencies and we solve optimization 
problem~\ref{eq:findingX} over real numbers, 
relaxing the integer constraint on $|D^t|$ set sizes. \emph{However, this 
result can serve as a powerful guideline in code generation. Furthermore, if 
derived tiling sizes generate a valid code, it is provably I/O optimal.}

\vspace{-0.5em}
\section{Projecting Programs onto SOAP}
\label{sec:beyondSOAP}
By the definition of SOAP, one input array may be accessed by different access 
function vector components, only if they form the simple overlap access --- 
that 
is, the accesses are offset by a constant stride. However, our analysis may go 
beyond this constraint if additional assumptions are met.

\subsection{Non-Overlapping Access Sets} 
Given input array $A$ and its access function components  $\bm{\phi}(\bm{\psi}) 
= [\bm{\phi}_1(\bm{\psi}_1), \dots, \bm{\phi}_n(\bm{\psi}_n)]$, if all access 
sets are disjoint, that is: $\forall_{i \ne j} \bm{\phi}_i[\mathcal{D}] \cap 
\bm{\phi}_j[\mathcal{D}] = \emptyset$, then we represent it as $n$ disjoint 
input arrays $A_i$ accessed by single corresponding access function component 
$\bm{\phi}_i(\bm{\psi}_i)$.

\begin{eg}
	\label{eg:lu}
	Consider the following code fragment from LU decomposition:
\end{eg}

\vspace{-0.5em}
\begin{lstlisting}[basicstyle=\ttfamily\footnotesize, mathescape=true]
for k in range(N):
  for i in range(k+1,N):
    for j in range(k+1,N):
$St:$       A[i,j] = A[i,j] - A[i,k] * A[k,j]
\end{lstlisting}
\emph{
	The analysis of iteration variables' domains $\mathcal{D}^i, \mathcal{D}^j, 
	\mathcal{D}^k$ shows that for fixed value of $k_0$, there are no two 
	iteration 
	vectors $\bm{\psi}_1 = [k_0, i_1, j_1]$ and $\bm{\psi}_2 = [k_0, i_2, j_2]$ 
	such that $[i_1, k_0] = [k_0, j_2] $ $\lor [i_1, j_1] = [k_0, j_2]$ $\lor 
	[i_1, 
	j_1] = [i_1, k_0]$, therefore, their access sets are disjoint. 
	Furthermore, for $k_0$, all elements from $A$ in range $[(k_0, N), 
	(k_0,N)]$ 
	are updated. Therefore, all accesses of form $[i_1, k_1] = [k_2, j_2]$ 
	access 
	\emph{different} vertices. We model this as a SOAP statement with three 
	disjoint arrays:}

\vspace{-0.5em}
$$St_2: A_1[i,j] = f(A_1[i,j], A_2[i,k], A_3[k,j])$$

\subsection{Equivalent Input-Output Accesses}
If array $A$ is updated by statement $St$ --- i.e., it is both input and output 
--- 
then we require that the output access function $\bm{\phi}_0$ is different 
than 
the input access function $\bm{\phi}_i$. If the input program does not meet 
this requirement, we can add additional ``version dimension'' to access 
functions that is offset by a constant between input and 
output accesses.

\begin{eg}
	Consider again Example~\ref{eg:lu}. Observe that array $A_1$ is updated (it 
	is 
	both the input and the output of $St_2$. Furthermore, both access functions 
	are 
	equal: $\bm{\phi}_0 = \bm{\phi}_1 = [i,j]$. We can associate a unique 
	version 
	(and therefore, a vertex) of each element of $A$ with a corresponding 
	iteration 
	of the $k$ loop. We add the version dimension associated with $k$ and 
	offset it 
	by constant $1$ between input and output:
\end{eg}

\vspace{-0.5em}
$$St_3: A_1[i,j,k+1] = f(A_1[i,j,k], A_2[i,k], A_3[k,j])$$

\subsection{Non-Injective Access Functions}
Given input array $A$ and its access function vector $\bm{\phi}$, we require 
that $\forall \bm{\psi}_i \ne \bm{\psi}_j: A[\bm{\phi}(\bm{\psi}_i)] \ne 
A[\bm{\phi}(\bm{\psi}_j)]$. If this is not the case, then we seek to bound the 
size of such overlap, that is, given subcomputation domain 
$\bm{D}(\mathcal{H})$, how many different iteration vectors $\bm{\psi}_j$ map 
to the same array element $A[\bm{\phi}(\bm{\psi}_i)]$. We can solve this by 
analyzing the iteration domain $\bm{\mathcal{D}}$ and the access function 
vector $\bm{\phi}$. 
If one array dimension is accessed by a function of multiple iteration 
variables $g(\phi^1, \dots, \phi^k)$ and $g$ is linear w.r.t. all $\phi^i$, 
the number of different values $g$ takes in $\bm{D}(\mathcal{H})$ is bounded 
by  $\max_{i = 1, \dots, k}\{|D^i|\} \le |g[\mathcal{H}]| \le \prod_{i 
=1}^{k}|D^i|$, for $ D^i \ne \{0\}, i = 1, \dots, k$.

\begin{eg}
	A single layer of  the direct convolution used in neural networks 
	may be 
	written as 
	seven nested loops with iteration variables $b,c,k,w,h,r,s$ and statement 
	(c.f.~\cite{demmelCNN}):
\end{eg}
\vspace{-1em}
$$St: Out[k,h,w,b] += Image[r + \sigma_w w, s + \sigma_h h, c, b]
\times Filter[k, r, s]$$
\emph{
	Depending on the value of $\sigma_w$ and $\sigma_h$, the access function 
	of  
	$Image$, $\bm{\phi} = [r + \sigma_w w, s + \sigma_h h, c, b]$ may not be 
	injective. Yet, observe that:}
\begin{enumerate}
	\item  $\sigma_w \ge |D^r| \land \sigma_h \ge |D^s| \implies \bm{\phi}$ is 
	injective $\implies $ $|\bm{\phi}[\mathcal{H}_{max}]| \ge 
	|D^r|\cdot|D^w|\cdot|D^s|\cdot|D^h|\cdot|D^c|\cdot|D^b|$
	\item $\sigma_w = 1 \land \sigma_h = 1 \implies 
	|\bm{\phi}[\mathcal{H}_{max}]| \ge$  \linebreak $\max(|D^r|,|D^w|)\cdot 
	\max(|D^s|,|D^h|)\cdot|D^c|\cdot|D^b|$,
\end{enumerate}
\emph{
	Our analysis provides a \emph{conditional} computational intensity:
	$\rho_{min} = \sqrt{S}/2$ in case (1) and $\rho_{max} = S/2$ in case (2). 
	Observe that
	case (2) yields the maximum non-injective overlap (maximum number of 
	different 
	iteration vectors map to the same element in $Image$). For any other values 
	of 
	$\sigma_w$ and $\sigma_h$, we have $\rho_{min} \le \rho \le \rho_{max}$. 
}

\section{Multi-Statement SOAP}
\label{sec:multistatement}
I/O lower bounds are not composable: the I/O cost of a program containing 
multiple statements may be lower than the sum of the I/O costs of each 
statement if evaluated in isolation. Data may be reused and merging 
of 
statements may lower the I/O cost.

Note that the number of vertices in the program's CDAG $G$ depend on 
domain sizes $D^i$ of each iteration variable. However, our derived 
upper bound of the computational intensity $\rho$ is 
\emph{independent} of 
the CDAG size, as it depends only on the access functions 
$\bm{\phi}_j$. This is also true for programs that contain 
multiple 
statements - to bound $\rho$ for multi-statement SOAP, we only need 
to model dependencies between the arrays and how they are accessed - 
e.g., one statement may take as an input an array that is an 
output 
of a different statement.

We represent the data flow between the program statements with a 
\emph{symbolic}
directed graph $G_S = (V_S, E_S)$.
For a given statement $St_i$, denote $In(St_i) = \{A_{i,1}, \dots, 
A_{i,m}\}$ a set of input arrays of statement $St_i$. Analogously, 
denote $Out(St_i)$ the set containing the output array of $St_i$. 
Analogously to 
program CDAG $G$ that captured dependencies between 
particular array elements, $G_S$ models dependencies 
between whole arrays (Figure~\ref{fig:soap_flow}).

\begin{defn} \textbf{Symbolic Digraph: SDG}
	\label{sec:sg}
	Given $k$-statement SOAP ${St_1, \dots, St_k}$, its 
	\emph{symbolic digraph (SDG)} $G_S = (V_S, E_S)$ 
	is a directed 
	graph where $V_S = \bigcup_{i=1}^k 
	(\text{In}(St_i) 
	\cup 
	\text{Out}(St_i))$ and $(A_u, A_v) \in E_S \iff \exists
	St_i : A_u \in In(St_i) \land A_v \in Out(St_i)$.
\end{defn} 

\noindent
$G_S$ is a directed graph, where vertices represent arrays accessed 
by a program, and edges represent data dependencies between 
them. Two arrays $A_u$ and $A_v$ are connected if there is a 
statement that accesses $A_u$ and computes $A_v$. Each edge 
is annotated with 
the 
corresponding access function vector of the statement that 
generates it.

\begin{eg}
	Consider the example in Figure~\ref{fig:soap_flow}. We have two statements 
	$St_1$ and $St_2$, with $\text{In}(St_1) = \{A, B\}$, $\text{Out}(St_1) = 
	\{C\}$, $\text{In}(St_2) = \{C, D, E\}$, $\text{Out}(St_2) = \{E\}$. We 
	then 
	construct the SDG $G_S = (V_S, E_S)$, with $V_S = \text{In}(St_1) \cup 
	\text{Out}(St_1) \cup \text{In}(St_2) \cup \text{Out}(St_2) = \{A, B, C, D, 
	E\}$. Furthermore, we have edges $E_S = \{(A,C), (B,C), (C,E), (D,E), 
	(E,E)\}$. The 
	edges are annotated with the corresponding access function vectors 
	$\bm{\phi}_{St1,1}, \dots, \bm{\phi}_{St2, 3}$.
\end{eg}

\textbf{Note: } While the ``explicit'' program CDAG $G=(V,E)$, where every 
vertex 
represents a single computation is indeed acyclic, the SDG $G_S = 
(V_S, E_S)$ may contain self-edges when a statement updates the loaded array 
($(E,E)$ in the example above). In $G$, one vertex corresponds to 
\emph{one version} of a single array element, while in $G_S$, one 
vertex encapsulates \emph{all versions} of all array elements.

\subsection{SDG Subgraphs}
\label{sec:sdg_subgraphs}
Denote $I \subset V_S$ set of input vertices of $G_S$ ($\forall A \in I: 
\text{indegree}(A)= 0$).
Let $H \subset V_S \setminus I$ be a subset of the vertices 
of SDG $G_S=(V_S, 
E_S)$. 
The SDG subgraph $G_S[H]$ is a 
subgraph of $G_S$ induced by the vertex set $H$. It 
corresponds to some subcomputation in which at least one 
vertex from each array in $H$ was computed. We now 
use the analogous strategy to the \xparting abstraction: 
since the optimal pebbling has an associated \xpart with 
certain properties (the dominator set constraint), we  bound the cost of any 
pebbling by finding the maximum 
subcomputation among \emph{all} valid $X$-partitions. We now 
show 
that every subcomputation in the optimal \xpart has a 
corresponding SDG 
subgraph $G_S[H]$. Therefore, finding $G_S[H_{opt}]$ that 
maximizes the computational intensity among \emph{all} 
subgraphs bounds the size of the maximal subcomputation 
(which, in turn, bounds the I/O cost of any pebbling).

Recall that an optimal pebbling $P$ has an associated \xpart 
$\mathcal{P}(X)$, where each $\mathcal{H} \in 
\mathcal{P}(X)$ represents a sequence of operations that 
are not interleaved with 
other subcomputations. Given $G_S$, each 
$\mathcal{H} \in \mathcal{P}(X)$ has an associated subgraph 
$G_S[H]$ s.t. every array vertex $A_i \in H$ represents an array 
from which at least one vertex was computed in $\mathcal{H}$.

Note that both the pebbling $P$ and the partition 
$\mathcal{P}(X)$ depend on the size of the CDAG that is 
determined by the sizes of the iteration domains $D^i$. 
However, the SDG does not depend on them. Thus, by finding 
the subgraph that maximizes the computational intensity, we 
bound $\rho$ for \emph{any} combination of input parameters. 

\begin{defn}
	The \textbf{subgraph SOAP statement} $St_H$ of subgraph 
	$G_S[H]$ is a single SOAP statement
	with the input  $In(St_H) = \{A 
	: A 
	\notin H \land \exists B \in H : (A,B) \in E_S\}$. Additionally, for each 
	vertex $B \in H$ that is not computed in $H$, that is 
	$\nexists A \in H: 
	(A,B) \in E_S$ , self-edges $(B,B) \in E$ are preserved ($B \in In(St_H)$).
\end{defn}

\noindent
\begin{Intuition}
The subgraph statement $St_H$ is a ``virtual'' SOAP statement 
that 
encapsulates multiple statements $St_1, \dots, St_k$. 
Given $H$, its subgraph statement's inputs $In(St_H)$ are 
formed 
by merging inputs $\bigcup_{i=1}^k In(St_i) \setminus V(H)$ 
from all statements that form $H$, but are not 
in $H$. By the construction of the SDG, this is 
equivalent to the definition 
above: take all vertices $A \in V_s \setminus V(H)$ that 
have a child in $V(H)$, that is $\exists B \in V(H) : (A,B) 
\in E_S$ (see Figure~\ref{fig:soap_flow}).

This forms the lower bound on the number of inputs for a 
corresponding subcomputation $\mathcal{H}$: all the vertices from 
arrays $A_i \in V(H)$ could potentially be computed during 
$\mathcal{H}$ and do not need to be loaded, but at least vertices 
from arrays $In(St_H)$ have to be accessed.
\end{Intuition}

\begin{eg}
	Consider again the example from Figure~\ref{fig:soap_flow}. 
	The set of input nodes is $I = \{A, B, D\}$. There are three 
	possible subgraph statements: $H_1 = \{C\}$, with $In(St_{H_1}) = 
	\{A,B\}$,  $H_2 = \{C\}$ with $In(St_{H_2}) = \{C,D,E\}$, and $H_3 = 
	\{C,E\}$ with $In(St_{H_3}) = \{A,B,D\}$. Note that by definition, the 
	self-edge $(C,C)$ is preserved in $H_2$, but not in $H_3$. Subgraphs $H_1$ 
	and $H_2$ 
	correspond to the input statements $St_1$ and $St_2$. Subgraph $H_3$ 
	encapsulates a subcomputation $\mathcal{H}$ that computes 
	some vertices 
	from both arrays $C$ and $E$, merging subcomputations $St_1$ and 
	$St_2$ and reusing outputs from $St_1$ to compute $E$.	
\end{eg}

Then, we 
establish the following lemma:

\begin{lma}
	\label{lma:sgpartition}
	Given an \xpart $\mathcal{P}(X) = \{\mathcal{H}_1, 
	\dots, \mathcal{H}_s\}$ of the $k$-statement SOAP, with its 
	corresponding 
	$G_S=(V_S,E_S)$, each subcomputation $\mathcal{H}$ has an 
	associated 
	intensity $\rho_\mathcal{H} = 
	\frac{|\mathcal{H}|}{|Dom_{min}(\mathcal{H})| - S}$ that 
	is 
	upper-bounded by the computational intensity of the subgraph 
	statement $St_H$ (Lemma~\ref{lma:rhoAPP}).
\end{lma}

\begin{proof}
	Recall that given the subcomputation $\mathcal{H}$, its 
	corresponding SDG subgraph $H$ is constructed as follows: 
	for each vertex $v \in V$ computed 
	during $\mathcal{H}$ belonging to some array $A_i$, add the 
	corresponding array vertex $s_i$ to $H$. Note that we allow 
	a vertex recomputation: if some vertex is (re)computed during the 
	optimal schedule of $\mathcal{H}$, its array vertex will 
	belong to $H$.
	
	Observe that by this construction and by definition of the 
	subgraph statement, all arrays from which at least one vertex is 
	loaded during $\mathcal{H}$ are in $In(St_H)$. Furthermore, 
	$In(St_H)$ is a subset of these arrays: during $\mathcal{H}$, 
	there might be some loaded vertex from array $A_j \in H$, but, 
	by definition of $St_H$, this array will not be in $In(St_H)$. 
	Therefore, $St_H$ lower bounds the input size of $\mathcal{H}$.
	
	The last step of the proof is to observe that by 
	Lemma~\ref{lma:rhoAPP}, the computational intensity of $St_H$ 
	bounds the maximum number of computed vertices for any 
	$\mathcal{H}' \in \mathcal{P}(X)$  that belong to $H$, 
	that 
	is, the union of all arrays in $H$. But since all vertices 
	that are computed in $\mathcal{H}$ belong to one of these 
	arrays, $\mathcal{H}$ cannot have higher computational 
	intensity.
\end{proof}

\vspace{-1.0em}
\subsection{SDG I/O Lower Bounds}
\label{sec:sdg_lowerbounds}

We now proceed to establish a method to derive the I/O lower bounds 
of the multi-statement SOAP given its SDG $G_S = (V_S, E_S)$.

For each array vertex $A \in V_S$, denote $|A|$ as the total 
number 
of vertices in the CDAG that belong to array $A$. Denote 
further 
$\mathcal{S}(A)$ the set of all subgraphs of $G_S$ that 
contain $A$. 
Then we prove the following theorem:

\begin{thm}
	\label{thm:sdg_lowerbound}
	The I/O cost $Q$ of a $k$-statement SOAP represented by the SDG 
	$G_S = (V_S, E_S)$ is bounded by
	\begin{equation}
	\label{eq:sdglowerbound}
	Q \ge \sum_{A \in V_S} \frac{|A|}{\max_{H \in 
	\mathcal{S}(A)} 
	\rho_{H}}
	\end{equation}
	\noindent
	where $\max_{H \in \mathcal{S}(A)} \rho_{H}$ is the 
	maximum 
	computational intensity over all subgraph statements of subgraphs 
	$H$ that contain vertex $A$.
\end{thm}

\begin{proof}
	This theorem is a direct consequence of 
	Lemma~\ref{lma:sgpartition} and the fact that all vertices in 
	CDAG $G$ are associated with some array vertex in SDG $G_S$. 
	Lemma~\ref{lma:sgpartition}, together with the definition of 
	$\mathcal{S}(a)$, states that $\max_{H \in 
	\mathcal{S}(a)} 
	\rho_{H}$ is the upper bound on any subcomputation 
	$\mathcal{H}$ that contains any vertex from array $a$. 
	Since 
	there are $|a|$ vertices associated with $a$, at least 
	$\frac{|a|}{\max_{H \in \mathcal{S}(a)} 
	\rho_{H}}$ I/O operations must be performed to compute these 
	vertices. Since the computational intensity expresses the average 
	cost \emph{per vertex}, even if some subcomputation in an optimal 
	\xpart spans more than one array, this is already modeled by the 
	set $\mathcal{S}(a)$. Therefore, we can sum the I/O costs 
	per 
	arrays $a$, yielding inequality~\ref{eq:sdglowerbound}. 
\end{proof}

Note that applying Theorem~\ref{thm:sdg_lowerbound} requires iterating over all 
possible subgraphs. In the worst case, this yields exponential complexity, 
prohibiting scaling our method to large programs. However, many scientific 
applications contain a limited number of kernels with simple dependencies. In 
practice we observed that our approach scales well to programs containing up to 
35 statements.

\section{Evaluation}
We evaluate our lower bound analysis on a wide range of applications, ranging 
from fundamental computational kernels and solvers to full workloads in 
hydrodynamics, numerical weather prediction, and deep learning. The set of 
applications covers both the previously analyzed kernels 
(the Polybench suite~\cite{polybench}, direct convolution), and kernels that were 
never analyzed before due to complicated dependency structures (multiple NN 
layers, diffusion, advection). Not only our tool covers broader class of 
programs than state-of-the-art approaches, but also it improves bounds 
generated by methods dedicated to specific narrower 
classes~\cite{olivry2020automated}. Improving I/O lower bounds has not only 
theoretical implications: loose bounds may not be applicable for generating 
corresponding parallel codes, as too many overapproximations may yield an 
 invalid schedule.

In our experiments we use DaCe~\cite{dace} to extract SOAP statements from 
Python and C code, and use MATLAB for symbolic analysis.
\iftr
; see Appendix A for 
implementation details.
\fi

\begin{table} 
	\begin{tabular}{l|lll} 
		\toprule 
		&  &  & 
		\textbf{Improv.}\\ 
		& \textbf{Kernel} & \textbf{SOAP I/O Bound} & 
		\textbf{over SotA}\\ 
		
		\midrule\multirow{32}{*}{\begin{turn}{90}\hspace{-3em}\textbf{Polybench\cite{olivry2020automated}}\end{turn}}&
		adi & $\frac{12 N^2 T}{\sqrt{S}}$ & 
		$\frac{12}{\sqrt{S}}$ 
		\\
		& atax & $M N$ & $1$\\ 
		& bicg & $M N$ & $1$\\ 
		& cholesky & $\frac{N^3}{3 \sqrt{S}}$ & $2$\\ 
		& correlation & $\frac{M^2 N}{\sqrt{S}}$ & $2$\\ 
		& covariance & $\frac{M^2 N}{\sqrt{S}}$ & $2$\\ 
		& deriche & $3 H W$ & $3$\\ 
		& doitgen & $\frac{2 N_P^2 N_Q N_R}{\sqrt{S}}$ & 
		$1$\\ 
		& durbin & $\frac{3 N^2}{2}$ & $3$\\ 
		& fdtd2d & $\frac{2 \sqrt{3} N_X 
		N_Y T}{\sqrt{S}}$ & $6\sqrt{6}$\\ 
		& floyd-warshall & $\frac{2 N^3}{\sqrt{S}}$ & $2$\\ 
		& gemm & $\frac{2 N^2}{\sqrt{S}}$ & $1$\\ 
		& gemver & $N^2$ & $1$\\ 
		& gesummv & $2 N^2$ & $1$\\ 
		& gramschmidt & $\frac{M N^2}{\sqrt{S}}$ & $1$\\ 
		& heat3d & $\frac{6 N^3 T}{\sqrt[3]{S}}$ & 
		$\frac{32}{3\sqrt[3]{3}}$\\ 
		& jacobi1d & $\frac{2 N T}{S}$ & $8$\\ 
		& jacobi2d & $\frac{4 N^2 T}{\sqrt{S}}$ & 
		$6\sqrt{3}$\\ 
		& 2mm & $\frac{4 N^3}{\sqrt{S}}$ & $1$\\ 
		& 3mm & $\frac{6 N^3}{\sqrt{S}}$ & $1$\\ 
		& lu & $\frac{2 N^3}{3 \sqrt{S}}$ & $1$\\ 
		& ludcmp & $\frac{2 N^3}{3 \sqrt{S}}$ & $1$\\ 
		& mvt & $N^2$ & $1$\\ 
		& nussinov & $\frac{N^3}{3 \sqrt{S}}$ & $2$\\ 
		& seidel2d & $\frac{4 N^2 T}{\sqrt{S}}$ & $6 
		\sqrt{3}$\\ 
		& symm & $\frac{2 M^2 N}{\sqrt{S}}$ & $1$\\ 
		& syr2k & $\frac{2 M N^2}{\sqrt{S}}$ & $2$\\ 
		& syrk & $\frac{M N^2}{\sqrt{S}}$ & $2$\\ 
		& trisolv & $\frac{N^2}{2}$ & $1$\\ 
		& trmm & $\frac{M^2 N}{\sqrt{S}}$ & $1$\\ 
\midrule\multirow{6}{*}{\begin{turn}{90}\textbf{Neural 
Networks}\end{turn}} 
& Direct conv. & $\frac{2 C_{in} C_{out} H_{out} N 
W_{out} W_{ker} 
H_{ker}}{\sqrt{S}}$ & 8 \\
& Softmax & $4 B H M N$ & ---\\
& MLP & $\frac{2 N (fc_{1} fc_{2} + fc_1 inp 
	+ fc_2 
	out)}{\sqrt{S}}$ & ---\\ 
& LeNet-5 & $\frac{300 \sqrt{2} C H N W}{\sqrt{S}}$ & 
---\\ 
& BERT Encoder & 
$\frac{4\,B\,H\,P\,L\,\left(L+2\,H\,P\right)}{\sqrt{S}}$ 
& ---\\ 
\midrule\multirow{3}{*}{\begin{turn}{90}\textbf{Various}\end{turn}} 
& LULESH & $22 \cdot \text{numElem}$ & ---\\
& horizontal diff. & $2 I J K$ & ---\\ 
& vertical adv. & $5 I J K$ & 
---\\ 
\bottomrule 
\end{tabular} 
	\caption{\small{Simplified leading-order terms of the I/O 
	lower 
	bounds 
	extracted from multi-statement SOAP and 
	previous state-of-the-art. For the direct convolution layer, the best 
	previously known bound was published by Zhang et 
	al.~\cite{chineseCNN1}.}
} 
	\label{tbl:bounds} 
\end{table}

\macsection{Polybench}
As our first case study, we analyze Polybench~\cite{polybench}, a polyhedral 
application benchmark suite composed of 30 programs from several domains, 
including 
linear algebra kernels, linear solvers, data mining, and computational biology. 
Prior best results were obtained by IOLB~\cite{olivry2020automated}, a tool 
specifically designed for analyzing I/O lower bounds of affine programs.
We summarize the results in Table~\ref{tbl:bounds}, listing the leading order 
term for brevity.

We find that SOAP analysis derives tight I/O lower bounds for 
all Polybench kernels.
Analyzing these programs as multi-statement 
SOAP either reproduces existing tight bounds, or improves 
them by constant 
factors (e.g., in Cholesky decomposition) on 14 out of 30 
applications (Table~\ref{tbl:bounds}). Of 
particular note is \texttt{adi} (Alternating Direction 
Implicit 
solver). Our algorithm detected a possible tiling in the time 
dimension, yielding the lower bound $(12N^2T)/\sqrt{S}$, 
compared 
to $N^2T$ reported by Olivry et al.~\cite{olivry2020automated}. 
However, due to dependency chains incurred by alternating 
directions, 
such tiling may 
violate loop-carried dependency constraints, which our algorithm 
relaxes. A parallel machine could potentially take advantage of this 
tiling scheme, possibly providing super-linear communication 
reduction. However, this is outside of the scope of this paper.

\macsection{Neural Networks} 
Analyzing I/O lower bounds of neural networks is a nascent field, and so far
only single-layer convolution was analyzed~\cite{demmelCNN,chineseCNN1}. 
We improve the previously-reported bound reported by 
Zhang et al.~\cite{chineseCNN1} by a factor of 8.

\subsection{New Lower Bounds}

Analyzing SOAP and the SDG representation enables capturing complex
data dependencies in programs with a large number of statements. To demonstrate 
this, we study 
larger 
programs in three fields, where no previous I/O bounds are 
known.
If an application contains both 
SOAP and data-dependent kernels, we find a SOAP 
representation that bounds the access sizes from below.

\macsection{LULESH}
The Livermore Unstructured Lagrangian Explicit Shock Hydrodynamics 
(LULESH)~\cite{lulesh} application is an unstructured phy-sics simulation. We analyze the main computational kernel, totaling over 60\% of runtime within one time-step of the simulation from the full C++ source code. As 
LULESH falls outside the purview of affine programs, this result is the 
first reported I/O lower bound.

\macsection{Numerical Weather Prediction}
We select two benchmark stencil applications from the COSMO Weather 
Model~\cite{cosmo} --- horizontal diffusion and vertical 
advection --- representatives of the two major workload types in the model's dynamical core.

\macsection{Deep Neural Networks}
For deep learning, we choose both individual representative operators (Convolution and Softmax) and network-scale 
benchmarks. Previous approaches only study data movement 
empirically~\cite{ivanov20}. To the best of our knowledge, we are the 
first to 
obtain 
I/O lower bounds for full networks, including a Multi-Layer Perceptron 
(MLP), the LeNet-5 CNN~\cite{lenet}, and a BERT Transformer 
encoder~\cite{transformer}.

\vspace{-0.5em}
\section{Related Work}

I/O analysis spans almost the entire history of general-purpose computer 
architectures, and graph pebbling abstractions were among the first methods to 
model memory requirements.
Dating back to challenges with the register allocation 
problem~\cite{completeRegisterProblems}, pebbles were also used to prove 
space-time tradeoffs~\cite{pebbleTradeoffs} and maximum parallel speedups by 
investigating circuit depths~\cite{dymond1985speedups}. Arguably the most 
influential pebbling abstraction work is the red-blue pebble game by Hong and 
Kung~\cite{redblue} that explicitly models load and store 
operations in a 
two-level-deep memory hierarchy. This work was extended numerous times, by: 
adding blocked access~\cite{externalMem}, multiple memory 
hierarchies~\cite{redblueHierarchy}, or introducing additional pebbles to allow CDAG 
compositions~\cite{redbluewhite}. Demaine and Liu proved that finding the 
optimal pebbling in a standard and no-deletion red-blue pebble game is PSPACE-complete~\cite{redbluehardSPAA}. Papp and Wattenhofer introduced a game variant 
with a non-zero computation cost and investigated pebbling approximation 
algorithms~\cite{papp2020hardness}.

Although the importance of data movement minimization is beyond doubt, the 
general solution for arbitrary algorithms is still an open problem. 
Therefore, many works were dedicated to investigate lower bounds only for 
single algorithms (often with accompanying implementations), 
like matrix-matrix multiplication~\cite{loomisApplied,2.5DLU,CARMA,COSMA}, 
LU~\cite{2.5DLU} and Cholesky decompositions~\cite{cholesky1, choleskyQRnew}. 
Ballard et al.~\cite{ballard2014communication} present an extensive collection 
of linear algebra algorithms.
Moreover, a large body of work exists for minimizing communication in irregular
algorithms~\cite{besta2015accelerating, sakr2020future}, such as Betweenness
Centrality~\cite{maciejBC}, min cuts~\cite{gianinazzi2018communication},
BFS~\cite{slimsell}, matchings~\cite{besta2020substream}, vertex similarity
coefficients~\cite{besta2020communication}, or general graph
computations~\cite{besta2017push, besta2021graphminesuite,
besta2021graphminesuite}.
Many of them use linear algebra based formulations~\cite{kepner2016mathematical}.
Recently, convolution networks gained high attention. The first asymptotic I/O 
lower bound for single-layer direct convolution was proved by Demmel et 
al.~\cite{demmelCNN}. Chen et al.~\cite{chineseCNN2} propose a matching 
implementation, and Zhang et al.~\cite{chineseCNN1} present the first 
non-asymptotic I/O lower bound for Winograd convolution.

In parallel with the development of I/O minimizing implementations for 
particular algorithms, several works investigated I/O lower bounds for whole 
classes of programs. Christ et al.~\cite{general_arrays} use a discrete version 
of Loomis-Whitney inequality to derive asymptotic lower bounds for 
single-statement programs nested in affine loops. Demmel and 
Rusciano~\cite{demmelHBL} extended this work and use discrete 
Hölder-Brascamp-Lieb inequalities to find optimal tilings for such programs.
{The polyhedral model}~\cite{polyhedralModel} {is widely used in practice
by many 
compilers}~\cite{pluto, polly}. {However, polyhedral methods 
have their own limitations: 1) they cannot capture 
non-affine loops}~\cite{loopcounting}; {2) while the representation of a 
program 
is polynomial, finding optimal transformations is still 
NP-hard}~\cite{loopFusionComplexity};{ 3) they are 
inapplicable for 
many neural network architectures, e.g., the Winograd algorithm for 
convolution}~\cite{chineseCNN1}. 
 
Recently, Olivry et al.~\cite{olivry2020automated} presented IOLB --- a tool 
for automatic derivation of non-parametric I/O lower bounds for programs 
that can be modeled by the polyhedral framework. IOLB employs both 
``geometric'' projection-based bounds based on the HBL 
inequality~\cite{demmel1}, as well as the wavefront-based approach from 
Elango~\cite{elango2014characterizing}. To the best of our knowledge, this is 
the only 
method that can handle multiple-statement programs. However, the IOLB model 
explicitly disallows recomputation that may be used to 
decrease the I/O cost, 
e.g., in the Winograd convolution algorithm, backpropagation, or vertical advection.
Furthermore, the framework is strictly limited to affine access programs. 
Even then, our method is able to improve those bounds by up to a factor of $6\sqrt{6}$ (\texttt{fdtd2d}) using a 
single, 
general method without the need to use application-specific techniques, such as wavefront-based reasoning.

\vspace{-0.5em}
\section{Conclusions}
In this work we introduce SOAP --- a broad class of statically 
analyzable programs. 
Using the explicit assumptions on the allowed overlap between 
arrays, we are able to precisely count the number of accessed 
vertices on the induced parametric CDAG. This stands in contrast with many state-of-the art 
approaches that are based on bounding projection sizes, as 
they need to underapproximate their union size, often 
resulting in a significant slack in constant factors of their 
bounds.
%
Our single method is 
able to reproduce or improve existing lower bounds for 
many important scientific kernels from various domains, 
ranging from 2$\times$ increase in the lower bound for linear 
algebra (\texttt{cholesky}, \texttt{syrk}), to more than 10$\times$ 
for stencil applications (\texttt{fdtd2d}, \texttt{heat3d}).

Our SDG abstraction precisely models data dependencies in 
multiple-statement programs. It directly captures input and 
output reuse, and allows data recomputation. Armed with these 
tools, we are the first to establish I/O lower bounds for 
entire neural networks, as well as core components of the popular Transformer architecture.

We believe that our work will be further extended to handle 
data-dependent accesses (e.g., sparse matrices), as well as 
scale better with input program size. The derived maximum 
subcomputation sizes can guide compiler optimizations and 
development of new communication-optimal algorithms through 
tiling, parallelization, or loop fusion transformations.

\vspace{-0.5em}
\section{Acknowledgements}
This project received funding from the European Research Council (ERC) under the European 
Union’s Horizon
2020 programme (grant agreement DAPP, No. 678880). Tal Ben-Nun is supported by the Swiss 
National Science Foundation (Ambizione Project \#185778). The authors wish to thank the 
Swiss National Supercomputing Center (CSCS) for providing computing infrastructure and 
support.

\bibliographystyle{IEEEtran}
\bibliography{refs}

\end{document}